\documentclass{article}
\usepackage[utf8]{inputenc}
\usepackage{mathtools}
\usepackage{amsfonts} 
\usepackage{amssymb}
\usepackage{appendix}
\usepackage{amsmath}
\usepackage{bbold}
\usepackage[pdftex]{graphicx}
\usepackage{tikz-cd}
\usepackage{mathrsfs}
\usepackage{titlesec}
 \usepackage{pstricks-add}
 \usepackage{comment}
 \usepackage{epsfig}
 \usepackage[nottoc,notlof,notlot]{tocbibind} 
\usepackage{pst-grad} % For gradients
\usepackage{pst-plot} % For axes
\usepackage[space]{grffile} % For spaces in paths
\usepackage{etoolbox} % For spaces in paths
\usepackage{hyperref}
\usepackage{color}
\usepackage{amsthm}
\hypersetup{
    colorlinks,
    citecolor=black,
    filecolor=black,
    linkcolor=black,
    urlcolor=black
}
\DeclareMathOperator{\ParallelTransportBars}{\mathit{\big\|}}

\newcommand{\VecF}[1]{\mathfrak{X}\left(#1\right)}
\newcommand{\PT}[4]{\ParallelTransportBars_{#1 \rightarrow #2}^{#3}#4}
\newcommand{\limucp}{\mathop{\lim}\limits_{\substack{\mathrm{ucp} \\ \epsilon \rightarrow 0}}}

\newcommand{\hittingtime}[1]{\tau^h_{#1}}
\newcommand{\exittime}[1]{\tau^e_{#1}}
\newcommand{\hittingtimeproc}[2]{\tau^h_{#1}\left(#2\right)}
\newcommand{\exittimeproc}[2]{\tau^e_{#1}\left(#2\right)}
\newcommand{\stopproc}[2]{#1^{|#2}}
\newcommand{\sem}[1]{\mathscr{S}\left(#1\right)}
\newcommand{\lifetime}[1]{\xi_{#1}}
\newcommand{\quadvar}[2]{\left\{#1, #2\right\}}
\newcommand{\Sint}[2]{\int #1\circ d #2}
\newcommand{\Iint}[2]{\int #1 d^I #2}
\newcommand{\Sinttime}[4]{\int_{#1}^{#2} #3\circ d #4}
\newcommand{\Iinttime}[4]{\int_{#1}^{#2} #3 d^I #4}
\newcommand{\interior}[1]{\mathrm{int\:}#1}
\newcommand{\dotp}[2]{\left\langle #1, #2 \right\rangle}
\newcommand{\Sinttimedotp}[4]{\int_{#1}^{#2}\dotp{#3}{\circ d #4}}
\newcommand{\ind}[1]{\mathbb{1}_{#1}}
\newcommand{\pont}[1]{\mathcal{P}#1}
\newcommand{\pontbun}[1]{T#1\oplus T^*#1}
\newcommand{\pard}[2]{\frac{\partial#1}{\partial#2}}
\newcommand{\Ac}[2]{\mathcal{S}_{#1}(#2)}
\newcommand{\Dual}[1]{#1^{\vee}}
\newcommand{\Strat}[2]{\mathrm{Strat}(#1, #2)}
\newcommand{\pr}[1]{\mathrm{pr}_{#1}}
\newcommand{\Proj}[1]{\mathrm{Pr}_{#1}}
\newcommand{\hitexittime}[1]{\tau^{(h,e)}_{#1}}
\begin{document}
\title{On Stochastic Variational Principles}
\author{Archishman Saha}
\date{}
\def\tsph{\mathbb{S}^2}
\def\E{\mathbb{E}}
\def\P{\mathcal{P}}
\def\L{\mathcal{L}}
\def\expec{\mathbb{E}}
\def\Rp{\mathbb{R}_+}
\def\lag{\mathfrak{g}}
\def\alag{\tilde{\mathfrak{g}}}
\def\dlag{\mathfrak{g}^*}
\def\Ver{\text{Ver}}
\def\Hor{\text{Hor}}
\def\dalag{\tilde{\mathfrak{g}}^*}
\def\ker{\text{ker}}
\def\im{\text{im}}
\def\sthm{(M,\{\cdot,\cdot\},h,X)}
\def\hamsec{\tilde{X}_H}
\def\poin{Poincar\'e\:}
\def\pb{\{,\cdot,\}}
\def\Cinc{C^{\infty}_c}
\def\R{\mathbb{R}}
\def\pman{(M, \{\cdot,\cdot\})}
\def\sman{(M, \Omega)}
\def\Bvec{\mathbf{B}}
\def\Jvec{\mathbf{J}}
\def\cotman{(T^*M, \Omega)}
\def\inv{^{-1}}
\def\tman{(TM, \Omega)}
\def\Ito{It\^o\:}
\def\d{\mathbf{d}}
\def\D{\mathcal{D}}
\def\qvec{\mathbf{q}}
\def\pvec{\mathbf{p}}
\def\lvec{\mathbf{\Lambda}}
\def\deldeleps{\frac{\partial}{\partial \epsilon}\Big|_{\epsilon = 0}}
\def\Jvec{\mathbf{J}}
\def\Avec{\mathbf{A}}
\def\dd{\mathbf{d}d}
\def\cotM{T^*M}
\def\cotQ{T^*Q}
\def\Qvec{\mathbf{Q}}
\def\Pvec{\mathbf{P}}
\def\Vvec{\mathbf{V}}
\def\uvec{\mathbf{u}}
\def\vvec{\mathbf{v}}
\def\Cin{C^{\infty}}
\def\Ck{C^k}
\def\F{\mathbb{F}}
\def\G{\mathcal{G}}
\def\E{\mathbb{E}}
\def\stin{\int\left<\Theta,\D X\right>}
\def\omsec{\tilde{\Omega}}
\def\omsym{\Omega^{\text{sym}}}
\def\Sym{\text{Sym}^2}
\def\stpM{\tau_p M}
\def\kepphase{T^*(\mathbf{R}^3\setminus\{0\})}
\def\Ad{\text{Ad}}
\def\ad{\text{ad}}
\def\scpM{\tau^*_p M}
\def\del{\circ d}
\def\pb{\{\cdot,\cdot\}}
\def\scM{\tau^* M}
\def\scN{\tau^* N}
\def\scqN{\tau^*_q N}
\def\stM{\tau M}
\def\lco{L\'{a}zaro-Cam\'{\i}\:and\:Ortega\:}
\def\Deldeleps{\frac{D}{D\epsilon}\Big|_{\epsilon = 0}}
\def\Zvec{\mathbf{Z}}
\def\stN{\tau N}
\def\dim{\mathrm{dim}\:}
\def\grad{\mathbf{\nabla}}
\def\covar{\delta^A}
\def\stqN{\tau_q N}
\def\grad{\mathbf{\nabla}}
\def\ostar{\textcircled{$\star$}}
\def\ostarthm{\textup{\textcircled{$\star$}}}
\def\Xvec{\textbf{X}}
\def\TxS{T^\times S^3_{np}}
\def\ToneS{T^{\times}_1S^3_{np}}
\def\M{\mathcal{M}}
\newtheorem{thm}{Theorem}[section]
\newtheorem{cor}{Corollary}[section]
\newtheorem{lem}{Lemma}[section]
\newtheorem{defn}{Definition}[section]
\newtheorem{example}{Example}[section]
\newtheorem{prop}{Proposition}[section]
\newtheorem{rem}{Remark}[section]
\newtheorem{theorem}{Theorem}[section]
\newtheorem{remark}{Remark}[section]
\newtheorem{proposition}{Proposition}[section]
\newtheorem{lemma}{Lemma}[section]
\newtheorem{corollary}{Corollary}[section]
\newtheorem{definition}{Definition}[section]
\newtheorem{conjecture}{Conjecture}[section]
\maketitle
\allowdisplaybreaks

\begin{abstract}
The study of stochastic variational principles involves the problem of constructing fixed-endpoint and adapted variations of semimartingales. We provide a detailed construction of variations of semimartingales that are not only fixed at deterministic endpoints, but also fixed at first entry times and first exit times for charts in a manifold. We prove a stochastic version of the fundamental lemma of calculus of variations in the context of these variations. Using this framework, we provide a generalization of the stochastic Hamilton-Pontryagin principle in local coordinates to arbitrary noise semimartingales. We also formulate a stochastic analogue of Noether's theorem in this context. For the corresponding global form of the stochastic Hamilton-Pontryagin principle, we introduce a novel approach to global variational principles by restricting to semimartingales obtained as solutions of Stratonovich equations. 
\end{abstract}

\section{Introduction}

Variational principles are ubiquitous in mechanical systems. At its heart, these principles involve finding a curve that extremizes an action integral among all curves with fixed endpoint conditions. While introducing noise in the framework of mechanics, one is naturally tempted to extend deterministic variational principles to the stochastic regime.  Two distinct kinds of stochastic variational principles exist in the literature: the first involves a stochastic action obtained by perturbing a deterministic Lagrangian by external noise; see, for instance, Street and Takao \cite{street2023}, \lco \cite{lco1}, Bou-Rabee and Owhadi \cite{bou2009stochastic}, Holm \cite{holm1}, Arnaudon et. al. \cite{holm2}, Cruzeiro et. al. \cite{holm3} and Crisan and Street \cite{crisan}.  The second is a deterministic action, evaluated by averaging a stochastic integral obtained from a single deterministic Lagrangian acting on stochastic paths. This viewpoint is present, for instance, in the works of Yasue \cite{yasue}, Cipriano and Cruzeiro \cite{cipriano2007}, Arnaudon and Cruzeiro \cite{anabela}, Arnaudon, Chen and Cruzeiro \cite{arn}, Zambrini \cite{zambrini1986variational} and Huang and Zambrini \cite{huangzambrini}. This also provides a probabilistic interpretation of Feynman's path integral approach to quantum mechanics and the reader is referred to Zambrini \cite{zambrini2014researchprogramstochasticdeformation} for more details on this.

\medskip

In this paper we will focus on the first viewpoint. Here the action is defined via a Stratonovich integral. While the Stratonovich integral behaves well under coordinate transformations, it poses some analytic difficulties. As remarked in Emery \cite{emerybook}, unlike \Ito integrals, a dominated convergence theorem is not true for Stratonovich integrals. Thus, in general, it is not true that if $(\Gamma_n)$ is a sequence of semimartingales that converge pointwise to a semimartingale $\Gamma$ and are dominated by a locally bounded process then the Stratonovich integrals $\int \Gamma_n\del X$ converges almost surely to $\int \Gamma\del X$ for any semimartingale $X$ uniformly on compacts in probability (ucp). This means that if $\Gamma$ and $X$ are real-valued semimartingales and $\{\Gamma_{\epsilon}\}$ is a family of semimartingales such that $\Gamma_{0,t} = \Gamma_t$ and the maps $\epsilon\mapsto \Gamma_{\epsilon,t}(\omega)$ are smooth for almost every sample point $\omega$, then one cannot conclude directly that the Stratonovich integral $\Sint{\frac{\Gamma_\epsilon - \Gamma}{\epsilon}}{X}$ converges ucp to $\Sint{\deldeleps\Gamma_{\epsilon}}{X}$. To overcome this difficulty, we assume differentiability in terms of the semimartingale topology as opposed to ucp topology. 

\medskip

Introducing stochasticity also leads to local and global difficulties in variational principles. The global issue involves fixing the final condition in the variational principle. In general, fixing a stochastic process to assume a certain distribution at a future time may lead to breakdown in adaptedness with respect to a given filtration. This leads to the problem of constructing fixed endpoint and adpated variations of a stochastic process $\Gamma$. 

\medskip

Broadly, two distinct solutions exist to this problem in the literature. They differ in the choice of the final time in the variational principle. The first involves fixing a compact set $K$ that contains the initial condition $\Gamma_0 = a$, for some point $a$ in the manifold, and fixing the final time to be the first exit time $\tau_K$ of $\Gamma$ from $K$. One then defines a vector field $X$ such that $X$ vanishes on $\{a\}\cup \partial K$ and constructs a variational family that yields the variation $\delta \Gamma = X(\Gamma)$. This ensures that $\delta \Gamma$ equals $0$ at $t = 0$ and at $t = \tau_K$. This approach is present in the works of \lco \cite{lco1} and Street and Takao \cite{street2023}. While suited to the stochastic environment, it is not  clear how this technique applies in the simpler deterministic set-up.

\medskip

The second approach involves constructing a variation by parallel transporting a deterministic curve $v(t)$ in the tangent space of $a$ such that $v(0) = v(T) = 0$ for some $T>0$, along the process $\Gamma$. This is used in case of geodesically complete manifolds, for instance, in the works of Arnaudon, Chen and Cruzeiro \cite{arn} for Lie groups, and Huang and Zambrini \cite{huangzambrini} for compact manifolds. The adaptedness of the variation is ensured by placing future conditions only on deterministic objects. Note that in this approach the final time is fixed and independent of the process. A related approach which considers a fixed final time is the use of Malliavin calculus. We refer the interested reader to Cruzeiro et. al. \cite{holm3} for more details. While this approach does not involve choosing a compact set and it accounts for a deterministic final time $T$, it involves more structure on the manifold.

\medskip

 Fixing the final time $t = T$ in the stochastic setting leads to local problems. To elaborate on this, first consider the case of deterministic Euler-Lagrange equations. Recall that the proof of the equivalence of Hamilton's principle and Euler-Lagrange equations in manifolds, using partial derivatives of the form $\frac{\partial L}{\partial q}$ and $\frac{\partial L}{\partial \dot{q}}$ (see, for instance, Marsden and Ratiu \cite{marsden2}, Theorem 8.1.3 and its proof) proceeds in local coordinates, by dividing the curve $q(t)$ into a finite number of segments, each of which lies in a chart. For any chart $U$ such that $U$ has a non-empty intersection with the curve $q(t)$, one can find a time interval $[t_0, t_1]$ such that $q(t)$ lies in $U$ for $t\in [t_0, t_1]$. Roughly speaking, global fixed endpoint problems, with fixed initial and final times, lead to local fixed endpoints problems with fixed inital and final times.

\medskip

But this is not the case if one considers a semimartingale $\Gamma$ instead of a deterministic curve. Spatial localization of a semimartingale in charts leads to temporal localization within stopping times. This means that if $U\subseteq M$ is a chart, the first hitting time of $\Gamma$ in $U$ and the first exit time from $U$ are (random) stopping times. Thus, globally considering a fixed endpoint problem with a final time $T>0$ does not lead to fixed endpoint problems locally with a fixed final time. Therefore, should we want to do local computations with fixed endpoint conditions it is necessary that we construct variations that vanish at $t= 0$, $t = T$, as well as at the first hitting time and the first exit time. The first two conditions are necessary since $U$ may contain the initial condition, or the exit time from $U$ may exceed $T$.

\medskip

A main objective of this paper is to introduce variations of semimartingales that vanish not only at initial and final deterministic times but also at the first hitting and exit times for a chart in a manifold. This allows us to do variational principles in local coordinates on a manifolds. 

\medskip

We also describe a novel method for working with variational principles globally on manifolds. We exploit the fact that Stratonovich equations on manifolds are determined by Stratonovich operators and these are deterministic generalizations of vector fields. By restricting the action integral to solutions of Stratonovich equations and assuming the noise to be an arbitrary semimartingale, we reformulate the problem of finding critical points of a stochastic action to determining a Stratonovich operator. Since the Stratonovich operator is a deterministic object, this suggests that the problem is solvable by deterministic arguments. We demonstrate that this is indeed the case, and in fact, this method allows us to bypass some of the complications that arise in the local case while working globally. 

\medskip

This paper is structured as follows: after reviewing some terminologies and notations from stochastic calculus in Section 2, we introduce variations of semimartingales in Section 3. We prove a stochastic analogue of the fundamental lemma of the calculus of variations as well for Stratonovich integrals, especially taking into account variations that vanish at the first hitting and exit times for a chart. In Section 4 we turn our focus on the stochastic Hamilton-Pontryagin principle. The stochastic Hamilton-Pontryagin principle was formulated by Bou-Rabee and Owhadi \cite{bou2009stochastic} and studied more recently by Street and Takao \cite{street2023}. As an application of the variational framework developed in Section 3, a proof of the local form of the stochastic Hamilton-Pontryagin principle is presented. This generalizes the Hamilton-Pontryagin principle formulated in \cite{street2023} to arbitrary noise semimartingales. We also discuss a stochastic version of Noether's theorem on the variational principle side. Then we discuss the intrinsic form of the stochastic Hamilton-Pontryagin principle by working at the level of Stratonovich operators.

\section{Notations and Terminologies from Stochastic Calculus}
We will always consider continuous semimartingales defined on a probability space $(\Omega, \mathcal{F}, P)$. If $S,T$ are predictable stopping times then we define 
\begin{align*}
    [[S,T]] &= \{(\omega, t)\in \Omega\times[0, \infty)\:|\:S(\omega)\leq t\leq T(\omega)\}\\
    [[S,T[[ &= \{(\omega, t)\in \Omega\times[0, \infty)\:|\:S(\omega)\leq t< T(\omega)\}\\
    ]]S,T]] &= \{(\omega, t)\in \Omega\times[0, \infty)\:|\:S(\omega)<t\leq T(\omega)\}\\
    ]]S,T[[ &= \{(\omega, t)\in \Omega\times[0, \infty)\:|\:S(\omega)< t< T(\omega)\}.
\end{align*}The set of semimartingales on a manifold $M$ will be denoted by $\sem{M}$. At times we will slightly abuse notation and write $\Gamma_t$ to refer to a semimartingale $\Gamma$. The lifetime of $\Gamma$ will be denoted by $\lifetime{\Gamma}$. For simplicity, unless otherwise mentioned, semimartingales will be assumed to have infinite lifetime. Given a semimartingale $\Gamma$ we let $\Gamma = \Gamma_0 + M_{\Gamma}+ A_{\Gamma}$ denote the Doob-Meyer decomposition of $\Gamma$, where $M_{\Gamma}$ is a local martingale and $A_{\Gamma}$ is a finite variation process. Given a Borel set $A$ in a manifold $M$ and a semimartingale $\Gamma$, we define the \textbf{first hitting time} for $A$ as the random variable \[\hittingtime{A}(\Gamma)(\omega) = \inf\{t\in [0, \infty]|\:\:\Gamma(\omega)\in A\}\]and the \textbf{first exit time} from $A$ as the random variable \[\exittime{A}(\Gamma)(\omega) = \inf\{t\in [0, \infty]|\:\:\Gamma(\omega)\notin A\}.\]$\hittingtime{A}(\Gamma)$ and $\exittime{A}(\Gamma)$ are stopping times and they are predictable stopping times if $A$ is a closed set (see, for example, Emery \cite{emerybook}). If the process $\Gamma$ is understood from context then we will use the notation $\hittingtime{A}$ and $\exittime{A}$ for first hitting time and first exit time, respectively. We will also define $\tau^{(h,e)}_A(\Gamma)$ or $\tau^{(h,e)}_A$ to be $\exittimeproc{A}{\Gamma_{t+\hittingtimeproc{A}{\Gamma}}}$. Assuming $A$ is closed, by definition of $\tau^{(h,e)}_A$, $\Gamma$ takes its values in $A$ in $[[\hittingtime{A},\tau^{(h,e)}_A + \hittingtime{A}]]$ and if $\tau_1$ and $\tau_2$ are stopping times such that $\Gamma$ takes its values in A in $[[\tau_1, \tau_2]]$, then $[[\tau_1, \tau_2]]\subseteq [[\hittingtime{A}, \tau^{(h,e)}_A + \hittingtime{A}]]$. Thus $\Gamma|_{[[\hittingtime{A},\tau^{(h,e)}_A + \hittingtime{A}]]}$ is the portion of $\Gamma$ contained in $A$. If $\tau$ is a stopping time, we define the stopped process $\stopproc{\Gamma}{\tau}$ by $\stopproc{\Gamma}{\tau}_t(\omega) = \Gamma_{t\wedge T(\omega)}(\omega)$, where $t\wedge T(\omega)$ denotes the minimum of $t$ and $T(\omega)$, for all $\omega\in \Omega$. 

\medskip

Given $\Gamma\in\sem{M}$ and a locally bounded predictable $\cotM$-valued process $Z$ over $\Gamma$ (that is, $Z$ projects to $\Gamma$), the Stratonovich integral of $Z$ along $\Gamma$ is denoted by $\Sint{Z}{\Gamma}$ and the \Ito integral of $Z$ along $\Gamma$ is denoted by $\Iint{Z}{\Gamma}$. If $\alpha$ is a 1-form on $M$ and $Z = \alpha(\Gamma)$ then $\Sint{\alpha}{\Gamma}:= \Sint{Z}{\Gamma}$ and $\Iint{\alpha}{\Gamma}:= \Iint{Z}{\Gamma}$. The reader is referred to Emery \cite{emerybook} for more details on Stratonovich and \Ito integrals.

\medskip

We will refer to Arnaudon and Thalamier \cite{arnaudon1998} for the topology of semimartingales on $\sem{M}$ and the topology of uniform convergence on compacts in probability (ucp) on the space $\mathbf{D}_c(M)$ of $M$-valued continuous, adapted processes. Endow $\Ck(M)$ with the topology of uniform convergence of compacts upto derivatives of order $k$. 
\begin{definition}
    Let $M, N$ be smooth manifolds and $E = \Ck(M)\times\mathbf{D}_c(M)$ or $\Ck(M)\times\sem{M}$ or $\mathbf{D}_c(M)$ or $\sem{M}$, and $F = \mathbf{D}_c(N)$ or $\sem{N}$. A map $\phi:E\rightarrow F$ is said to be \textbf{lower semicontinuous} if for every sequence $(X_n)$ in $E$ converging to $X\in E$, the sequence $\stopproc{\phi(X_n)}{\lifetime{\phi(X)}-}$ converges to $\phi(X)$
\end{definition}
\begin{remark}
    If we assume that semimartingales have infinite lifetime then $\phi$ is lower semicontinuous if and only if $\phi$ is continuous.
\end{remark}
\noindent For the proof of the next lemma we refer to Proposition 2.6 in \cite{arnaudon1998}.
\begin{lemma}\label{lowersemicontinuitylemma}
    \begin{enumerate}
        \item The map \begin{align*}
    \Ck(M)\times\mathbf{D}_c(M)  &\to \mathbf{D}_c(\R)  \\
    (f,X) &\mapsto f(X)
    \end{align*} is lower semicontinous.
    \item The maps \begin{align*}
    \Ck(M)\times\sem{M}  &\to \sem{\R}  \\
    (f,X) &\mapsto f(X),
    \end{align*}where $k\geq 2$, and
    \begin{align*}
        \sem{\R^n}&\to\sem{\R}\\
        X = (X^1, \cdots, X^n)&\mapsto M_{X^i} \:\mathrm{or\:} A_{X^i}\:\mathrm{or\:} \{M_{X^i}, M_{X^j}\},
    \end{align*}
    where $\{\cdot,\cdot\}$ denotes the quadratic variation, is lower semicontinous, are lower semicontinuous.
    \end{enumerate}
\end{lemma}
\begin{corollary}\label{convergenceofproducts}
     Let $\{X_n\}$ and $\{Y_n\}$ be two sequences of real-valued semimartingales. Suppose $X_n\rightarrow X$ and $Y_n \rightarrow Y$ in $\sem{\R}$, where $X$ and $Y$ are semimartingales. Then $X_nY_n\rightarrow XY$ in $\sem{\R}$.
\end{corollary}
\begin{proof}
    Let $\phi:C^0(\R^2)\times \sem{\R^2}\rightarrow \sem{\R}$ denote the map $(h,Z)\mapsto h(Z)$. Then $\phi$ is lower semicontinuous. Define $h\in C^0(\R^2)$ by $h(x,y) = xy$ and set $h_n = h$, $Z_n = (X_n, Y_n)$ and $Z = (X,Y)$. Since $X_n\rightarrow X$ and $Y_n \rightarrow Y$ in $\sem{\R}$, it follows that $Z_n \rightarrow Z$ in $\sem{\R^2}$. Then $(h_n,Z_n)$ converges to $(h,Z)$ in $C^0(\R^2)\times \sem{\R}$ and hence $h_n(Z_n) = X_nY_n\rightarrow h(Z) = XY$ in $\sem{\R}$. This completes the proof.
\end{proof}
\noindent The next proposition also follows from Lemma \ref{lowersemicontinuitylemma}. 
\begin{proposition}\label{convergencestrat}
    Let $(Z_n)$ be a sequence in $\sem{\R}$ that converges to $Z\in \sem{\R}$ and is dominated by a locally bounded predictable process $K$. If $X\in\sem{\R}$ then $\Sint{Z_n}{X}\xrightarrow{ucp} \Sint{Z}{X}$. 
\end{proposition}
\begin{proof}
    By definition $\Sint{Z_n}{X} = \Iint{Z_n}{X}+ \quadvar{Z_n}{X}$. The first term converges to $\Iint{Z}{X}$ in ucp by the \Ito dominated convergence theorem (see Emery \cite{emerybook}). Since convergence in the semimartingale topology implies convergence in ucp (see Arnaudon and Thalmaier \cite{arnaudon1998}), it suffices to show that $\quadvar{Z_n}{X}\rightarrow \quadvar{Z}{X}$ in $\sem{\R}$. Note that $\quadvar{Z_n}{X} = \quadvar{M_{Z_n}}{M_X}$. Let $Y_n = (M_{Z_n}, M_{X})\in \sem{\R^2}$. By Lemma \ref{lowersemicontinuitylemma}, $M_{Z_n}\rightarrow M_Z$ in $\sem{\R}$ and hence $Y_n \rightarrow Y:= (M_Z, M_X)$ in $\sem{\R^2}$. Again by Lemma \ref{lowersemicontinuitylemma}, we see that $\quadvar{Z_n}{X}=\quadvar{M_{Z_n}}{M_X}\rightarrow \quadvar{M_Z}{M_X}= \quadvar{Z}{X}$ in $\sem{\R}$. This completes the proof.
\end{proof}
We also recall the definition of a Stratonovich operator and a Stratonovich equation from Emery \cite{emerybook}.
\begin{definition}
    Let $N$ and $M$ be smooth manifolds. 
    \begin{enumerate}
        \item A \textbf{Stratonovich operator} $S$ from $N$ to $M$ is a family of linear maps \[\left\{S(x,y):T_xN\to T_yM\:|\:x\in N,\:y\in M\right\}\]smoothly depending on $x$ and $y$. The set of Stratonovich operators from $N$ to $M$ will be denoted by $\Strat{N}{M}$.
        \item Given $S\in\Strat{N}{M}$ and a semimartingale $X$ on $N$ a \textbf{solution} of the \textbf{Stratonovich equation}
        \begin{equation}\label{strateq}
            \del \Gamma = S(X,\Gamma)\del X
        \end{equation}
        is a semimartingale $\Gamma$ in $M$ that satisfies 
        \begin{equation}\label{strateqsol}
            \int \alpha \del \Gamma = \int S^{\vee}(X,\Gamma)\alpha\del X
        \end{equation}
        for every 1-form $\alpha$ on $M$. Here $\Dual{S}(x,y): T^*_yM\rightarrow T^*_xN$ denotes the dual of the linear map $S(x,y)$. If we want to explicitly refer to the semimartingale $X$ then we will denote the solution of \eqref{strateq} by $\Gamma_X$.
    \end{enumerate}
\end{definition}
We refer the reader to \cite{emerybook} for further details on Stratonovich equations, and in particular, for positive results on existence and uniqueness of solutions. 

\medskip

Given $x\in N, v\in T_xN$ and $S\in\Strat{N}{M}$ we obtain a vector field $S^{x,v}$ on $M$ given by $S^{x,v}(y) = S(x,y)(v)$. On the other hand let $V_1, \cdots, V_n$ are vector fields on $M$. Let $(e_1,\cdots, e_n)$ denote any basis of $\R^n$ define $S\in\Strat{\R^n}{M}$ by setting $S(x,y)(v^1, \cdots, v^n) = \sum_{i = 1}^nv^iV_i(y)$ for all $x\in\R^n, y\in M$ and $(v^1\cdots, v^n)\in \R^n \cong T_x\R^n$. Then $S^{x,e_i} = V_i$ for every $i = 1, \cdots, n$. 

\section{Variations of a Semimartingale}
In this section we describe variations of a semimartingale $\Gamma$ in a smooth manifold $M$. 
\begin{definition}\label{definitionofvariations}Let $\Gamma$ be a semimartingale in a smooth manifold $M$. A \textbf{deformation} of $\Gamma$ is a map $[-s,s]\rightarrow \sem{M}$ denoted by $\epsilon\mapsto \Gamma_{\epsilon}$, where $\epsilon\in [-s,s]$ for some $s>0$, such that:\begin{itemize}\item $\Gamma_{\epsilon}$ is a semimartingale for all $\epsilon>0$.\item $\Gamma_{0,t} = \Gamma_t$.\item The map $\epsilon \mapsto \Gamma_{\epsilon,t}$ is smooth for almost every path of $\Gamma$. Additionally, there exists a $TM$-valued semimartingale $\delta \Gamma$ such that for every $f\in\Cin(M)$, $\frac{f(\Gamma_{\epsilon}) - f(\Gamma)}{\epsilon} \rightarrow df(\delta \Gamma)$ in $\sem{\R}$ as $\epsilon\rightarrow 0$. The semimartingale $\delta\Gamma$ will be called a \textbf{variation} of $\Gamma$.\end{itemize}\end{definition}
\begin{remark}
    It is assumed implicitly that the lifetimes of the semimartingales $\Gamma_{\epsilon}$, are at least as large as the lifetime of $\Gamma$. 
\end{remark}
\begin{remark}
    Using Definition 2.9 in Arnaudon and Thalmaier \cite{arnaudon1998}, the definition of $\delta\Gamma$ implies that $\Gamma_{\epsilon}$ converges to $\delta \Gamma$ with respect to the semimartingale topology on $M$. 
\end{remark}
\begin{definition}\label{admissibility}
    Let $M$ be a smooth manifold and $\Gamma$ be a semimartingale in $M$. We say that $\Gamma$ is \textbf{admissible} if, for every semimartingale $Y$ in $TM$ over $\Gamma$, there exists a deformation $\epsilon\mapsto \Gamma_{\epsilon}$ of $\Gamma$ with $\delta \Gamma = Y.$
\end{definition}
\begin{theorem}
    Assume that $\Gamma$ is a semimartingale in a Riemannian manifold $M$ and $\exp$ denote the exponential map on $M$. If $\exp_{\Gamma_{t}(\omega)}$ has domain $T_{\Gamma_t(\omega)}M$ for all $t\geq 0$ and $\omega\in \Omega$ then $\Gamma$ is admissible.
\end{theorem}
\begin{proof}
    This follows directly from Corollary 4.3 in Arnaudon and Thalmaier \cite{arnaudon1998}. We remark that the hypothesis ensures that the lifetimes of $\Gamma_{\epsilon}$ are at least as large as the lifetime of $\Gamma$.
\end{proof}
\begin{remark}
    Using Hopf-Rinow theorem we conclude that if $M$ is connected and $M$ is a compact manifold or a geodesically complete manifold then every semimartingale $\Gamma$ on $M$ is admissible. 
\end{remark}
\begin{definition}
    Let $\Gamma\in\sem{M}$, $X\in \sem{\R}$, $f\in \Cin(M)$ and $\alpha$ be a 1-form on $M$. Given a deformation $\epsilon\mapsto \Gamma_{\epsilon}$, we define:
    \begin{enumerate}
        \item $\D\Sint{f(\Gamma)}{X} = \limucp \Sint{\frac{f(\Gamma_{\epsilon}) - f(\Gamma)}{\epsilon}}{X_t}$.
        \item $\D\Sint{\alpha}{\alpha} = \limucp \frac{1}{\epsilon}\left(\Sint{\alpha}{\Gamma_{\epsilon}} - \Sint{\alpha}{\Gamma}\right)$.
    \end{enumerate}
\end{definition}
\begin{remark}
    The notation $\D$ is used as opposed to $\delta$ in order to distinguish between ucp convergence and semimartingale convergence.
\end{remark}
The next lemma prescribes a method for computing variations of Stratonovich integrals.
\begin{lemma}\label{variationsofintegrals}
Let $\Gamma$ be a semimartingale in a manifold $M$ and $\epsilon\mapsto \Gamma_{\epsilon}$ be a deformation of $\Gamma$.
\begin{enumerate}
    \item For every real semimartingale $X$ and $f\in \Cin(M)$
    \begin{equation}\label{variationoffdx}
        \D\int f(\Gamma)\del X = \int df(\delta \Gamma)\del X
    \end{equation}
    \item For every 1-form $\alpha$ on $M$
    \begin{equation}\label{variationofalphadgamma}
        \D\int \alpha(\Gamma)\del \Gamma = \int i_{\delta \Gamma}\d\alpha \del X + \langle\alpha(\Gamma),\delta \Gamma\rangle + \langle\alpha(\Gamma_0), \delta\Gamma_0 \rangle,
    \end{equation}
    where $\d\alpha$ denotes the exterior derivative of $\alpha$.
\end{enumerate}
\end{lemma}
\begin{proof}
        The first statement follows by applying Proposition \ref{convergencestrat} to $Z_{\epsilon} := \frac{f(\Gamma_{ \epsilon}) - f(\Gamma)}{\epsilon}$. The proof of the second statement follows \textit{mutatis mutandis} from the proof of Proposition 4.3 by replacing Lemma 5.2 and $\stopproc{\Gamma}{\tau_K}$ therein with therein with Corollary \ref{convergenceofproducts} and $\Gamma$ respectively, and recalling that convergence in the semimartingale topology implies ucp convergence.
\end{proof}
\subsection{Fixed Endpoint Variations}
We will assume that $\Gamma$ is an admissible semimartingale in $M$. Let $T>0$ be fixed. Suppose $g\in\Cin(\R)$ is supported on $(0,T)$ and $X\in\VecF{M}$. Then $Y_t = g(t)X(\Gamma_t)$ is a semimartingale in $TM$ over $\Gamma$ (that is, the projection of $Y$ on $M$ is $\Gamma$) that vanishes at $t= 0$ and $t =T$. Then there exists a deformation $\epsilon\mapsto \Gamma_\epsilon$ of $\Gamma$ such that $\delta \Gamma = Y$. 

\medskip

A second way to construct variations that vanish at $t = 0$ and $t = T$ is inspired from the works of Arnaudon, Chen and Cruzeiro \cite{arn} and Huang and Zambrini \cite{huangzambrini}. Assume that $M$ is equipped with a connection, $\Gamma_0 = a$ for some $a\in M$ and let $\PT{0}{t}{\Gamma}{v}$ denote the parallel transport of a vector $v\in T_aM$ along $\Gamma$. Let $v(t)$ be a deterministic curve in $T_aM$ such that $v(0) = v(T) = 0$. Then $Y_t := \PT{0}{t}{\Gamma}{v(t)}$ is the $TM$-valued semimartingale over $\Gamma$ such that $Y_0 = Y_T = 0$. The admissibility hypothesis ensures that there exists a deformation $\epsilon\mapsto \Gamma_\epsilon$ of $\Gamma$ with $\delta \Gamma = Y$.

\medskip

Next, given a closed subset $K\subseteq M$ we describe how to construct variations of the portion of $\Gamma$ contained in $K$. Recall that $\hittingtime{K}$ is the hitting time for $K$ and if $\tau^{(h,e)}_K:= \exittimeproc{K}{\Gamma_{t+\hittingtimeproc{K}{\Gamma}}}$ then $\Gamma|_{[[\hittingtime{K}, \hittingtime{K}+ \tau^{(h,e)}_K]]}$ is the portion of $\Gamma$ that lies in $K$. Let $f\in\Cin(M)$ be supported on the interior $\interior{K}$ of $K$ and $X\in\VecF{M}$. Then $f\cdot X$ vanishes outside $\interior{K}$. It follows that $\tilde{Y} = f\cdot X(\Gamma)$ vanishes on $[[0, \infty[[\:\setminus\: ]]\hittingtime{K}, \hittingtime{K}+ \tau^{(h,e)}_K[[ = [[0, \hittingtime{K}]]\bigcup [[\hittingtime{K}+\tau^{(h,e)}_K,\infty[[$. Let $g\in\Cin(\R)$ be supported on $(0,T)$. Then $Y_t := g(t)\tilde{Y}_t$ is a $TM$ valued semimartingale that not only vanishes on $[[0, \hittingtime{K}]]\bigcup [[\hittingtime{K}+\tau^{(h,e)}_K,\infty[[$, but also for all $t \geq T$. The admissibility hypothesis shows that there exists a deformation $\epsilon\mapsto \Gamma_{\epsilon}$ of $\Gamma$ with $\delta \Gamma = Y$.
\begin{definition}
    Let $K\subseteq M$ be a closed subset and $\Gamma$ be an admissible semimartingale in $M$. 
    \begin{enumerate}
    \item A \textbf{$K$-deformation} of $\Gamma$ is a deformation $\epsilon\mapsto\Gamma_{\epsilon}$ of $\Gamma$ such that $\delta \Gamma$ vanishes outside $]]\hittingtime{K}, \hittingtime{K}+ \tau^{(h,e)}_K[[$. The corresponding variation will be called a \textbf{$K$-variation}. 
    \item Given $T>0$, a \textbf{$(K,T)$-deformation} of $\Gamma$ is a $K$-deformation $\epsilon \mapsto \Gamma_{\epsilon}$ of $\Gamma$ such that $\delta \Gamma$ vanishes on $[[\left(\hittingtime{K}+\tau^{(h,e)}_K\right)\wedge T,\infty[[$. The associated variation $\delta \Gamma$ will be called a \textbf{$(K,T)$-variation}.
    \end{enumerate}
\end{definition}
\begin{lemma}\label{kvariationlemma}
    Let $\epsilon\mapsto \Gamma_{\epsilon}$ be a $K$-deformation of $\Gamma$, where $K\subset M$ is closed. Then:
    \begin{enumerate}
        \item For every $f\in \Cin(M)$
        \[\D\Sinttime{0}{T}{f(\Gamma)}{X} = \Sinttime{\hittingtime{K}}{\hittingtime{K}+\tau^{(h,e)}_K}{df(\stopproc{\delta\Gamma}{T})}{X}\].
        \item For every 1-form $\alpha$ on $M$
        \[\D\Sinttime{0}{T}{\alpha}{\Gamma} = \Sinttime{\hittingtime{K}}{\hittingtime{K}+\tau^{(h,e)}_K}{i_{\delta\Gamma^{|T}}\d\alpha}{\Gamma^{|T}} + \dotp{\alpha(\Gamma_T)}{\delta\Gamma_T} - \dotp{\alpha(\Gamma_0)}{\delta\Gamma_0}.\]
    \end{enumerate}
\end{lemma}
\begin{proof}
    We only prove (1) since the proof of (2) is similar. It follows from the definition that $\delta \Gamma$ vanishes outside $]]\hittingtime{K}, \hittingtime{K}+ \tau^{(h,e)}_K[[$. Let $\mathbb{1}_{(\cdot)}$ denotes the indicator function. Using Proposition 5.3 in \lco \cite{lco1}, we have
    \begin{align*}\D\Sinttime{0}{T}{f(\Gamma)}{X} &=  \Sinttime{0}{T}{df(\delta\Gamma)}{X}\\&=\Sint{\mathbb{1}_{[0,T]}\mathbb{1}_{[[\hittingtime{K}, \hittingtime{K}+ \tau^{(h,e)}_K]]}df(\delta\Gamma)}{X}\\
    &=\Sint{\mathbb{1}_{[[\hittingtime{K}, \hittingtime{K}+ \tau^{(h,e)}_K]]}df(\delta\Gamma^{|T})}{X}\\
    &=\Sinttime{\hittingtime{K}}{\hittingtime{K}+\tau^{(h,e)}_K}{df(\stopproc{\delta\Gamma}{T})}{X}.
    \end{align*}
\end{proof}

\subsection{A Stochastic Analogue of the Fundamental Lemma of the Calculus of Variations}
We will formulate a stochastic analogue of the fundamental lemma of the calculus of variations in coordinate charts. Let $\dotp{\cdot}{\cdot}$ denote the standard Euclidean inner product and $(e_1, \cdots, e_n)$ denote the standard basis of $\R^n$.
\begin{lemma}\label{fundlem}
    Let $M$ be a smooth $n$-manifold and $U\subseteq M$ be a coordinate chart. We identify $U$ with an open subset of $\R^n$, also denoted by $U$. Let $\Gamma\in\sem{M}$ be admissible and $\Xi:\sem{M}\rightarrow\sem{\R^n}$ satisfy $\Xi(\Gamma)_{A_t} = \Xi(\Gamma_{A_t})$ for any continuous change of time $t\mapsto A_t$. If for every $(\bar{U},T)$-deformation $\epsilon\mapsto \Gamma_{\epsilon}$ we have
\[\int_{\hittingtime{\bar{U}}}^{\hittingtime{\bar{U}}+\tau^{(h,e)}_{\bar{U}}}\dotp{\delta\Gamma^{|T}}{\del \Xi(\stopproc{\Gamma}{T})} = 0,\]
    then $\del\Xi(\stopproc{\Gamma}{T}) = 0$ in $]]\hittingtime{U}, \hittingtime{U}+\tau^{(h,e)}_{U}[[$. Here $\del\Xi(\stopproc{\Gamma}{T}) = 0$ means that $\Xi(\Gamma^{|T}) - \Xi(\Gamma^{|T})_{\hittingtime{U}} = 0$ a.s. in $]]\hittingtime{U}, \hittingtime{U}+\tau^{(h,e)}_{U}[[$.
\end{lemma}
\begin{proof}
    First suppose $U$ is a precompact coordinate ball and identify $U$ with the open ball $B_r(0)$ of radius $r$ and centered at $0$. Given $s>0$ let $(g_n)$ be a sequence in $\Cin(\R)$ such that $g_n$ is supported in $(0, s+1)$ for every $n$ and $g_n\rightarrow\mathbb{1}_{(0,s]}$ pointwise. Then $g_n(t)\delta\Gamma^{|T}_t$ is a $(\bar{U},T)$ variation of $\Gamma$. Using the fact that $g_n$ is of bounded variation, the \Ito dominated convergence theorem and Proposition 5.3 of Lázaro-Camí and Ortega \cite{lco1}, we obtain
    \begin{align*}
    0&=\int_{\hittingtime{\bar{U}}}^{\hittingtime{\bar{U}}+\tau^{(h,e)}_{\bar{U}}}\dotp{g_n(t)\delta\Gamma^{|T}}{\del \Xi(\stopproc{\Gamma}{T})}\\
    &= \int_{\hittingtime{\bar{U}}}^{\hittingtime{\bar{U}}+\tau^{(h,e)}_{\bar{U}}}g_n(t)\del\left(\int \dotp{\delta\Gamma^{|T}}{\del \Xi(\Gamma^{|T})}\right)\\
    &= \int_{\hittingtime{\bar{U}}}^{\hittingtime{\bar{U}}+\tau^{(h,e)}_{\bar{U}}}g_n(t)d^I\left(\int \dotp{\delta\Gamma^{|T}}{\del \Xi(\Gamma^{|T})}\right)\\
    &\xrightarrow[n\rightarrow\infty]{ucp} \int_{\hittingtime{\bar{U}}}^{\hittingtime{\bar{U}}+\tau^{(h,e)}_{\bar{U}}}\ind{(0,s]}d^I\left(\int \dotp{\delta\Gamma^{|T}}{\del \Xi(\Gamma^{|T})}\right)\\
    &= \Iint{\ind{(0,s]}\ind{[[\hittingtime{\bar{U}}, \hittingtime{\bar{U}+ \hitexittime{\bar{U}}}]]}}{\left(\int \dotp{\delta \Gamma^{|T}}{\del \Xi(\Gamma^{|T})}\right)}\\
    &= \int_{0}^{s}\ind{[[\hittingtime{\bar{U}}, \hittingtime{\bar{U}+ \hitexittime{\bar{U}}}]]}\dotp{\delta \Gamma^{|T}}{\del \Xi(\Gamma^{|T})}.
    \end{align*}
    Since this holds for all $s>0$ we conclude that \[\int\ind{[[\hittingtime{\bar{U}}, \hittingtime{\bar{U}+ \hitexittime{\bar{U}}}]]}\dotp{\delta\Gamma^{|T}}{\del \Xi(\stopproc{\Gamma}{T})} = 0.\]Moreover, since $\delta\Gamma^{|T}$ vanishes outside $[[\hittingtime{\bar{U}}, \hittingtime{\bar{U}}+ \hitexittime{\bar{U}}]]$, we conclude that \[\int\dotp{\delta\Gamma^{|T}}{\del \Xi(\stopproc{\Gamma}{T})} = 0\]
    in $[[\hittingtime{\bar{U}}, \hittingtime{\bar{U}}+ \tau^{(h,e)}_{\bar{U}}]].$

    \medskip

    Let $0<\eta<r$ and $h\in\Cin(\R^n)$ be supported on $B_r(0)$ with $h|_{\bar{B}_{\eta}(0)} = 1$, where $\bar{B}_{\eta}(0)$ denotes the closed ball of radius $\eta$ centered at $0$. For $j = 1,\cdots,n$, let $\tilde{X}$ denote the vector field on $\R$ defined by $\tilde{X}_j = h(x)e_j$. Then $\tilde{X}_j$ vanishes on the boundary $\partial \bar{B}_r(0)$ of $\bar{B}_r(0) = \bar{U}$. Extending $\tilde{X}_j|_{\bar{U}}$ to a vector field $X$ on $M$ by setting $X = 0$ outside $U$ and letting $g\in\Cin(M)$ be supported on $(0,T)$, we can construct a $(\bar{U}, T)$-deformation $\epsilon\mapsto \Gamma_{\epsilon}$ with variation $g(t)X(\Gamma_t)$ by the admissibility hypothesis.

\medskip

Since $\int\dotp{\delta\Gamma^{|T}}{\del \Xi(\stopproc{\Gamma}{T})} = 0$ in $[[\hittingtime{\bar{U}}, \hittingtime{\bar{U}}+ \tau^{(h,e)}_{\bar{U}}]]$ and $[[\hittingtime{\bar{B}_{\eta}(0)}, \hittingtime{\bar{B}_{\eta}(0)}+ \tau^{(h,e)}_{\bar{B}_{\eta}(0)}]]\subseteq [[\hittingtime{\bar{U}}, \hittingtime{\bar{U}}+ \tau^{(h,e)}_{\bar{U}}]]$, it follows that $\int\dotp{\delta\Gamma^{|T}}{\del \Xi(\stopproc{\Gamma}{T})} = 0$ in $[[\hittingtime{\bar{B}_{\eta}(0)}, \hittingtime{\bar{B}_{\eta}(0)}+ \tau^{(h,e)}_{\bar{B}_{\eta}(0)}]]$. Let $Z_j$ denote the $j$th component of $\Xi(\Gamma^{|T})$. Then the previous equality and the fact that $X(\Gamma^{|T}) = g(t)e_j$ on $[[\hittingtime{\bar{B}_{\eta}(0)}, \hittingtime{\bar{B}_{\eta}(0)}+ \tau^{(h,e)}_{\bar{B}_{\eta}(0)}]]$ implies that
\[\int\ind{[[\hittingtime{\bar{B}_{\eta}(0)}, \hittingtime{\bar{B}_{\eta}(0)}+ \tau^{(h,e)}_{\bar{B}_{\eta}(0)}]]}g(t)\del Z_j = 0\]
for all $g\in\Cin(\R)$ supported on $(0,T)$.

\medskip

Pick an arbitrary $s\in(0,T)$. Replacing $g$ by $\tilde{g}_n$ where $(\tilde{g}_n)$ is a sequence in $\Cin(\R)$ that is supported on $(0,T)$ and $\tilde{g}_n\rightarrow \ind{(0,s]}$, we get
\begin{align*}
    0 &= \Sint{\ind{[[\hittingtime{\bar{B}_{\eta}(0)}, \hittingtime{\bar{B}_{\eta}(0)}+ \tau^{(h,e)}_{\bar{B}_{\eta}(0)}]]}\tilde{g}_n(t)}{Z_{j_t}}\\
    &=\Iint{\ind{[[\hittingtime{\bar{B}_{\eta}(0)}, \hittingtime{\bar{B}_{\eta}(0)}+ \tau^{(h,e)}_{\bar{B}_{\eta}(0)}]]}\tilde{g}_n(t)}{Z_{j_t}}\\
    &\xrightarrow[n\rightarrow\infty]{ucp}\Iint{\ind{[[\hittingtime{\bar{B}_{\eta}(0)}, \hittingtime{\bar{B}_{\eta}(0)}+ \tau^{(h,e)}_{\bar{B}_{\eta}(0)}]]}\ind{(0,s]}}{Z_{j_t}}\\
    &=\Iinttime{\hittingtime{\bar{B}_{\eta}(0)}}{\hittingtime{\bar{B}_{\eta}(0)}+\hitexittime{\bar{B}_{\eta}(0)}}{\ind{(0,s]}}{Z_j}\\
    &= \Iinttime{\hittingtime{\bar{B}_{\eta}(0)}}{\hittingtime{\bar{B}_{\eta}(0)}+\hitexittime{\bar{B}_{\eta}(0)}}{}{Z_j^{|s}}\\
    &= \left(Z_{j_{\left(\hittingtime{\bar{B}_{\eta}(0)}+\hitexittime{\bar{B}_{\eta}(0)}\right)\wedge s}} - Z_{j_{\hittingtime{\bar{B}_{\eta}(0)}}}\right)
\end{align*}
and we used the fact that $\tilde{g}_n$ is of bounded variation, the \Ito dominated convergence theorem and Proposition 5.3 in \cite{lco1}. Since this is true for all $0<s<T$, it follows that $\del Z_j = 0$ on $[[\hittingtime{\bar{B}_{\eta}(0)}, \hittingtime{\bar{B}_{\eta}(0)}+ \tau^{(h,e)}_{\bar{B}_{\eta}(0)}]]$. Since this holds for all $0<\eta<r$, we conclude that $\del Z_j = 0$ on $]]\hittingtime{B_r(0)}, \hittingtime{B_r(0)}+\tau^{(h,e)}_{B_{r}(0)}[[ \:=\: ]]\hittingtime{U}, \hittingtime{U}+\tau^{(h,e)}_U[[$. Consequently $\del\Xi(\Gamma^{|T}) = 0$ on $]]\hittingtime{U}, \hittingtime{U}+\tau^{(h,e)}_U[[$.

\medskip

Now suppose $U\subseteq M$ is a coordinate chart in $M$. For every precompact coordinate ball $U_0$ in $U$, note that a $(\bar{U}_0, T)$-deformation of $\Gamma$ is also a $(\bar{U}, T)$-deformation of $\Gamma$. By our hypothesis, for every $(\bar{U}_0, T)$-deformation of $\Gamma$, we have
\[\int_{\hittingtime{\bar{U}}}^{\hittingtime{\bar{U}}+\tau^{(h,e)}_{\bar{U}}}\dotp{\delta\Gamma^{|T}}{\del \Xi(\stopproc{\Gamma}{T})} = \int_{\hittingtime{\bar{U}_0}}^{\hittingtime{\bar{U}_0}+\tau^{(h,e)}_{\bar{U}_0}}\dotp{\delta\Gamma^{|T}}{\del \Xi(\stopproc{\Gamma}{T})}= 0.\]
This implies that $\del\Xi(\Gamma^{|T}) = 0$ in $]]\hittingtime{{U}_0}, \hittingtime{{U}_0}+\tau^{(h,e)}_{{U}_0}[[$. Since this holds for all precompact coordinate balls $U_0\subseteq U$, we have $\del\Xi(\Gamma^{|T}) = 0$ in $]]\hittingtime{{U}}, \hittingtime{{U}}+\tau^{(h,e)}_{{U}}[[$.
\end{proof}

\section{The Stochastic Hamilton-Pontryagin Principle}

The deterministic Hamilton-Pontryagin principle was introduced by Yoshimura and Marsden \cite{yoshimura06}. A stochastic extension was first introduced by Bou-Rabee and Owhadi \cite{bou2009stochastic} and has been generalized more recently by Street and Takao \cite{street2023}. We will provide a proof of the local form of the stochastic Hamilton-Pontryagin principle as an application of the variational framework developed in the last section to stochastic geometric mechanics. Then we will develop a stochastic version of Noether's theorem. This will be followed by a discussion of the intrinsic form of the stochastic Hamilton-Pontryagin principle, where we will use Stratonovich operators to provide a global description.

\medskip

Given a configuration manifold $Q$ let $\pont{Q}:= \pontbun{Q}$ denote its Pontryagin bundle. Local coordinates on $\pont{Q}$ will be denoted by $(q,v,p)$. Let $\L\in\Cin(TQ)$ be a Lagrangian. The deterministic Hamilton-Pontryagin principle states a $\pont{Q}$-valued curve $(q(t),v(t),p(t))$ is a critical point of the action \[\int_{t_0}^{t_1} [L(q(t), v(t)) + \dotp {p(t)}{ \dot{q}(t) - v(t)}]dt\]amongst all curves such that $q(t_0)$ and $q(t_1)$ are fixed, if and only if $(q(t), v(t), p(t))$ satisfies the implicit Euler-Lagrange equations given by \[\dot{q} = v, \:\:\:p = \pard{L}{v},\:\:\:\dot{p} = \pard{L}{q}.\]The reader is referred to \cite{yoshimura06} for more details on the deterministic Hamilton-Pontryagin principle and its application to constrained systems.

\medskip

\begin{definition}
    Let $X = (X^0, \cdots, X^k)\in\sem{\R^{k+1}}$ and $\L\in\Cin(\pont{Q})$ be a Lagrangian. Suppose we have $L_1, \cdots, L_k\in\Cin(Q)$ and vector fields $V_1, \cdots, V_k$ on $Q$. Given an admissible $\pont{Q}$-valued semimartingale $\Gamma_t = (q_t, v_t, p_t)$ we define the \textbf{Hamilton-Pontryagin action integral} as 
    \newpage
    \begin{align}\label{hampontlocal}
        \Ac{X}{\Gamma} = &\int_0^T\left(\L(q_t, v_t)\del X^0_t +\sum_{i = 1}^k L_i(q_t)\del X^i_t \right.\nonumber\\ &\left.+ \dotp{p_t}{\del q_t - v_t\del X^0_t - \sum_{i = 1}^k V_i(q_t)\del X^i_t}\right).
    \end{align}
\end{definition}
In Bou-Rabee and Owhadi \cite{bou2009stochastic} the authors consider $X = (t, B^i_t, \cdots, B^k_t)$, where $B^i$ is a Brownian motion, and $V_i = 0$. Street and Takao \cite{street2023} have generalized this to the case where $X$ is a driving semimartingale and $\Gamma$ is compatible with $X$. The reader is referred to \cite{street2023} as well as Street and Crisan \cite{crisan} for further details on driving semimartingales and the compatibility hypothesis, as well as a different stochastic analogue of the fundamental lemma of the calculus of variations under these assumptions. We will only assume that $X\in\sem{\R^{k+1}}$ and in particular, we will forego the assumption that $X^0 = t$.

\medskip

Let us also describe the action function intrinsically. For this we first recall some key ingredients involved in the deterministic setup from Yoshimura and Marsden \cite{yoshimura06}. Let $G:\pont{Q}\rightarrow \R$ denote the fibrewise pairing map between $TQ$ and $\cotQ$, that is, $G(q,v,p) = \dotp{p}{v}$. Denote by
\begin{align}\label{mapsinstrinsic}
    \Proj{\P Q} &:T\P Q\rightarrow \P Q\nonumber\\
    \Proj{T\P Q} &: TT\P Q\rightarrow T\P Q\nonumber\\
    \pr{Q} &:\P Q\rightarrow Q\nonumber\\
    \pr{TQ} &:\P Q\rightarrow TQ\nonumber\\
    \pr{\cotQ} &:\P Q\rightarrow \cotQ
\end{align}
the corresponding projection maps. We define the map $\rho_{T\cotQ}: T\cotQ\rightarrow \pont{Q}$ in local coordinates by setting $\rho_{T\cotQ}(q,p,v_q,v_p) = (q,v_q,p)$. An intrinsic definition can be found in Yoshimura and Marsden \cite{yoshimura061}. Let $\G$ denote the 1-form on $\pont{Q}$ given by $\G = G\circ\rho_{T\cotQ}\circ T\pr{\cotQ}$. In local coordinates, if $(u_q, u_v, u_p)\in T_{(q,v,p)}\P Q$ then
\begin{equation}\label{Ginlocalcoordinates}
    \G(q,v,p)(u_q, u_v, u_p) = G(q,u_q,p) = \dotp{p}{u_q}.
\end{equation}
Consequently, if $\Gamma_t = (q_t, v_t, p_t)$ then $\Sint{\G}{\Gamma} = \int \dotp{p_t}{\del q_t}$. 

\medskip

Given a vector field $V\in\VecF{Q}$ define $\tilde{V}:\P Q\rightarrow \P Q$ by $\tilde{V}(x) = (V\circ \pr{Q}(x))\oplus \pr{\cotQ}(x)\in \P Q$. Written in local coordinates this reads $\tilde{V}(q,v,p) = (q, V(q), p)$. For every $j\in \{0, \cdots, k\}$ define the \textbf{generalized energy} $E_j: \P Q\rightarrow \R$ by \[
E_j =
\begin{cases} 
G - \L \circ \pr{TQ}, & \text{if } j = 0, \\[8pt]
G \circ \tilde{V}_j - L_j \circ \pr{Q}, & \text{if } j = 1, \dots, k.
\end{cases}
\]
In coordinates, $E_0(q, v, p) = \langle p, v\rangle - L(q, v)$ and $E_i(q,v,p) = \langle p, V_i(q) \rangle - \L_i(q)$, for $i = 1, \cdots, k$. The generalized energies $E_i$ for $i = 1, \cdots, k$ also appear in Street and Takao \cite{street2023}. We note that if $\Gamma_t = (q_t, v_t, p_t)$ in local coordinates then 
\[
        E_j(\Gamma_t) = E_j(q_t, v_t, p_t) = \begin{cases}
            \langle p_t, v_t \rangle - \L(q_t, v_t, p_t), & \text{if } j = 0, \\[8pt]
\langle p_t, V_j(q_t)\rangle - L_j(q_t), & \text{if } j = 1, \dots, k.
        \end{cases}
\]
Hence 
\begin{equation}\label{hampontglobal}
    \Ac{X}{\Gamma} = \Sinttime{0}{T}{\G}{\Gamma} - \sum_{j = 0}^k \Sinttime{0}{T}{E_j(\Gamma)}{X_j}.
\end{equation}

\subsection{The Local Form of the Stochastic Hamilton-Pontryagin Principle}
First we describe variations in local coordinates of the terms in the Hamilton-Pontryagin action integral.
\begin{lemma}\label{variationinG}
    Let $\Gamma$ be an admissible semimartingale on $\P Q$ and let $\Gamma_t = (q_t,v_t,p_t)$ in local coordinates. Suppose $\epsilon\mapsto \Gamma_{\epsilon,t} = (q_{\epsilon,t}, v_{\epsilon,t}, p_{\epsilon,t})$ be a deformation of $\Gamma$. Then\begin{align*}\D\Sinttime{0}{T}{\G}{\Gamma}&=\D\int_0^T\langle p_t, \del q_t\rangle \\&= \int_0^{T}\dotp{\delta p_t^{|T}}{\del q_t^{|T}} - \int_{0}^{T}\dotp{\del \stopproc{p_t}{T}}{\delta \stopproc{q_t}{T}}\\&+\dotp{p_T}{ \delta q_T} - \dotp{p_0}{ \delta q_0}.\end{align*}
\end{lemma}
\begin{proof}
    In local coordinates $\dotp{\G(\Gamma_t)}{\delta \Gamma_t} = 
    \left\langle p_t, \delta q_t\right\rangle$. Then, by Lemma \ref{variationsofintegrals}
    \begin{align*}\D\Sinttime{0}{T}{\G}{\Gamma} &= \D\Sinttime{0}{T}{\G}{\Gamma^{|T}}\\&=\int_0^{T}i_{\delta \Gamma^{|T}}\d \G\del \Gamma^{|T}+\dotp{\G(\Gamma_T)}{\delta \Gamma_T} - \dotp{\G(\Gamma_0)}{\delta \Gamma_0}\\
    &=\int_0^{T}i_{\delta \Gamma^{|T}}\d \G\del \Gamma^{|T} + \dotp{p_T}{ \delta q_T} - \dotp{p_0}{ \delta q_0}.\end{align*}
    Let $(q,v,p)\in \P Q$. Suppose $(\dot{q}, \dot{v}, \dot{p}),(w_q, w_v, w_p)\in T_{(q,v,p)}\P Q$. From the local coordinate expression of $\G$ in Eq. \eqref{Ginlocalcoordinates}, it follows that $\G(q,v,p) = pdq$. Hence \[\d\G(q,v,p) = \sum_{i = 1}^{\dim Q}dp_i\wedge dq^i,\]which yields
    \[i_{(w_q, w_v, w_p)}\d\G(\dot{q}, \dot{v}, \dot{p}) = \d\G(q,v,p)((w_q, w_v,w_p), (\dot{q}, \dot{v}, \dot{p})) = \dotp{w_p}{\dot{q}} - \dotp{w_q}{\dot{p}}.\]
    Consequently
\begin{align*}\int_0^Ti_{\delta \Gamma^{|T}}\d \G\del \Gamma^{|T} =&\int_0^T\dotp{\delta p_t^{|T}}{\del q_t^{|T}} - \int_0^T\dotp{\del \stopproc{p_t}{T}}{\delta \stopproc{q_t}{T}}. \end{align*}This concludes the proof.
\end{proof}
\begin{remark}
    The product rule is often used to prove the above lemma. Mimicking the product rule we write
\[\D \int_0^{T} \langle p_t^{|T}, \del q_t^{|T}\rangle = \int_0^{T}\dotp{\delta p_t^{|T}}{\del q_t^{|T}} + \int_0^{T}\dotp{ \stopproc{p_t}{T}}{\delta \left(\del\stopproc{q_t}{T}\right)}.\]
But the term $\delta \left(\del\stopproc{q_t}{T}\right)$ is not defined since $\del\stopproc{q_t}{T}$ is not a stochastic process. To define this, we recall that the $\delta q_t^{|T}$ is assumed to be a semimartingale by definition. Hence we can set 
\[\int_0^T\dotp{ \stopproc{p_t}{T}}{\delta \left(\del\stopproc{q_t}{T}\right)} = \int_0^T\dotp{ \stopproc{p_t}{T}}{\del \left(\delta\stopproc{q_t}{T}\right)}.\]
The Stratonovich product rule gives us
\[\del \dotp{p_t^{|T}}{\delta q_t^{|T}} = \dotp{\del p_t^{|T}}{\delta q_t^{|T}}+ \dotp{p_t^{|T}}{\del \left(\delta q_t^{|T}\right)}\]
which implies
\begin{align*}
\dotp{p_T}{ \delta q_T} - \dotp{p_0}{ \delta q_0}=\int_0^T \dotp{\del p_t^{|T}}{\delta q_t^{|T}} + \int_0^T \dotp{p_t^{|T}}{\del \left(\delta q_t^{|T}\right)}.\end{align*}
Therefore\begin{align*}\D\int_0^T\langle p_t, \del q_t\rangle &= \int_0^{T}\dotp{\delta p_t^{|T}}{\del q_t^{|T}} - \int_{0}^{T}\dotp{\del \stopproc{p_t}{T}}{\delta \stopproc{q_t}{T}}\\&+\dotp{p_T}{ \delta q_T} - \dotp{p_0}{ \delta q_0}.\end{align*}\end{remark}
\begin{lemma}
    Let $\Gamma$ be an admissible semimartingale on $\P Q$ and let $\Gamma_t = (q_t,v_t,p_t)$ in local coordinates. Suppose $\epsilon\mapsto \Gamma_{\epsilon,t} = (q_{\epsilon,t}, v_{\epsilon,t}, p_{\epsilon,t})$ be a deformation of $\Gamma$. Then
    \begin{align*}
        &-\D\sum_{j = 0}^k \Sinttime{0}{T}{E_j(\Gamma)}{X^j}\\
    &=\D\left[\int_0^T\left(\left(\L(q_t, v_t) - \dotp{p_t}{v_t}\right)\del X^0_t + \sum_{i = 1}^k\left(L_i(q_t) - \dotp{p_t}{V_i(q_t)}\right)\del X^i_t\right) \right]\\ &= \int_0^T\frac{\partial}{\partial q_t^{|T}}\left\langle\left(\L\del X_t^0 + \sum_{i = 1}^k\left(L_i - \dotp{p_t^{|T}}{V_i(q_t^{|T})}\right)\circ dX^i_t\right), \delta q_t^{|T}\right\rangle\\
        &+\int_0^T\dotp{\left(\stopproc{p_t}{T} - \pard{\L}{\stopproc{v_t}{T}}\right)\del X^0_t}{\delta v_t^{|T}} \\&- \int_0^T \dotp{\delta\stopproc{p_t}{T}}{\stopproc{v_t}{T}\del X^0_t + \sum_{i = 1}^kV_i(\stopproc{q_t}{T})\del X^i_t}.
    \end{align*}
\end{lemma}
\begin{proof}
    By Lemma \ref{variationsofintegrals} we have
    \begin{align*}
        &- \D\sum_{j = 0}^k \Sinttime{0}{T}{E_j(\Gamma)}{X^j}\\
        &= - \D\sum_{j = 0}^k \int_0^T\left(\dotp{\pard{E_j}{q_t^{|T}}}{\delta q_t^{|T}} + \dotp{\pard{E_j}{v_t^{|T}}}{\delta v_t^{|T}} + \dotp{\pard{E_j}{p_t^{|T}}}{\delta p_t^{|T}}\right)\del X^j_t\\
        &=\int_0^T\left(\dotp{\pard{\L}{q_t^{|T}}}{\delta q_t^{|T}} + \dotp{\pard{\L}{v_t^{|T}}}{\delta v_t^{|T}} - \dotp{\delta p_t^{|T}}{v_t^{|T}} - \dotp{p_t^{|T}}{\delta v_t^{|T}}\right)\del X^0_t\\
        &+\sum_{i = 1}^k\int_0^T\left(\dotp{\pard{}{ q_t^{|T}}\left(L_i(q_t^{|T}) - \dotp{p_t^{|T}}{V_i(q_t^{|T})}\right)}{\delta q_t^{|T}}\right.-\left.\dotp{\delta p_t^{|T}}{V_i(q_t^{|T})}\right)  \del X^i_t\\
        &=\int_0^T\frac{\partial}{\partial q_t^{|T}}\left\langle\left(\L\del X_t^0 + \sum_{i = 1}^k\left(L_i - \dotp{p_t^{|T}}{V_i(q_t^{|T})}\right)\circ dX^i_t\right), \delta q_t^{|T}\right\rangle\\
        &+\int_0^T\dotp{\left(\stopproc{p_t}{T} - \pard{\L}{\stopproc{v_t}{T}}\right)\del X^0_t}{\delta v_t^{|T}} \\&- \int_0^T \dotp{\delta\stopproc{p_t}{T}}{\stopproc{v_t}{T}\del X^0_t + \sum_{i = 1}^kV_i(\stopproc{q_t}{T})\del X^i_t}.
    \end{align*}
\end{proof}
The local form of the stochastic Hamilton-Pontryagin principle is given by the following theorem:
\begin{theorem}\label{localhamiltonpontryagin}
    For every semimartingale $X = (X^0, \cdots, X^k)$ on $\R^{k+1}$, if $\Gamma_t = (q_t, v_t, p_t)\in \sem{\P Q}$ is admissible then $\D \Ac{X}{\Gamma} = 0$ for all deformations $\epsilon\mapsto \Gamma_{\epsilon}$ such that $\delta q_t = T\pr{Q}(\Gamma_t)$ vanishes at $t = 0$ and $t= T$ if and only if $\stopproc{\Gamma}{T} = \left(q_t^{|T}, v_t^{|T}, p_t^{|T}\right)$ the \textbf{stochastic implicit Euler-Lagrange equations} given by
    \begin{align}\label{sellocal}
    \circ dq_t &= v_t\circ dX^0_t + \sum_{i=1}^kV_i(q_t)\circ dX^i_t\nonumber\\
    \circ dp_t &= \frac{\partial}{\partial q_t}\left(\L\del X_t^0 + \sum_{i = 1}^k\left(L_i - \dotp{p_t}{V_i(q_t)}\right)\circ dX^i_t\right)\nonumber\\
    \left(p_t - \pard{\L}{v_t}\right)\del X^0_t&=0.
\end{align}
\end{theorem}
\begin{proof}
Let $\epsilon\mapsto \Gamma_{\epsilon,t} = (q_{\epsilon,t}, v_{\epsilon,t}, p_{\epsilon,t})$ be a deformation of $\Gamma$ such that $\delta q_t = 0$ at $t = 0$ and $t = T$. Consequently $\dotp{p_t}{\delta q_t} = \dotp{\G(\Gamma_t)}{\delta \Gamma_t}$ vanishes at $t = 0$ and $t = T$. Using Lemma \ref{variationsofintegrals}
\begin{align*}
    \D\Ac{X}{\Gamma} &= \D\left[\int_0^T \G\del \Gamma - \sum_{j = 0}^k\int_0^TE_j(\Gamma)\del X_j\right]\nonumber\\
    &= \int_0^T i_{\delta \Gamma^{|T}}\mathbf{d}\G\del \Gamma^{|T} - \sum_{j = 0}^k\int_0^T dE_j(\delta \Gamma^{|T})\del X_j \nonumber\\&+ \dotp{\G(\Gamma_T)}{\delta\Gamma_T} -  \dotp{\G(\Gamma_0)}{\delta\Gamma_0}\nonumber\\
    &=\int_0^T i_{\delta \Gamma^{|T}}\mathbf{d}\G\del \Gamma^{|T} - \sum_{j = 0}^k\int_0^T dE_j(\delta \Gamma^{|T})\del X_j \nonumber.
\end{align*}
Suppose $\Gamma$ solves the stochastic implicit Euler-Lagrange equations. We use the previous two lemmas to show that the above expression vanishes in local coordinates. We have
\begin{align}\label{calculationinlocalproof}
    &\int_0^T i_{\delta \Gamma^{|T}}\mathbf{d}\G\del \Gamma^{|T} - \sum_{j = 0}^k\int_0^T dE_j(\delta \Gamma^{|T})\del X_j \nonumber\\
    &=\int_0^T\left\langle\frac{\partial}{\partial q_t^{|T}}\left(\L\del X_t^0 \right.\right.\nonumber\\&+\left.\left. \sum_{i = 1}^k\left(L_i - \dotp{p_t^{|T}}{V_i(q_t^{|T})}\right)\circ dX^i_t\right) - \del \stopproc{p_t}{T}, \delta q_t^{|T}\right\rangle\nonumber\\
    &+ \int_0^T\dotp{\left(\stopproc{p_t}{T} - \pard{\L}{\stopproc{v_t}{T}}\right)\del X^0_t}{\delta v_t^{|T}}\nonumber\\
    &+\int_0^T \dotp{\delta\stopproc{p_t}{T}}{\del\stopproc{q_t}{T}-\stopproc{v_t}{T}\del X^0_t - \sum_{i = 1}^kV_i(\stopproc{q_t}{T})\del X^i_t}\nonumber\\
    &= \int_0^T\left\langle\left(\delta q_t^{|T}, \delta v_t^{|T}, \delta p_t^{|T}\right),\left(\int\left(\frac{\partial}{\partial q_t^{|T}}\left(\L\del X_t^0 +\right.\right.\right.\right.\nonumber\\ &\left.\left.\sum_{i = 1}^k\left(L_i - \dotp{p_t^{|T}}{V_i(q_t^{|T})}\right)\circ dX^i_t\right) - \del \stopproc{p_t}{T}\right),\Sint{\left(p_t^{|T} - \pard{\L}{v_t^{|T}}\right)}{X^0_t},\nonumber\\ &\left.\int\left(\del q_t^{|T} - v_t^{|T}\del X^0_t - \sum_{i = 1}^kV_i(q_t^{|T})\del X^i_t\right)\right\rangle\nonumber\\
    &= \int_0^T\left\langle\delta \Gamma_t^{|T},\left(\int\left(\frac{\partial}{\partial q_t^{|T}}\left(\L\del X_t^0 +\right.\right.\right.\right.\nonumber\\ &\left.\left.\sum_{i = 1}^k\left(L_i - \dotp{p_t^{|T}}{V_i(q_t^{|T})}\right)\circ dX^i_t\right) - \del \stopproc{p_t}{T}\right),\Sint{\left(p_t^{|T} - \pard{\L}{v_t^{|T}}\right)}{X^0_t},\nonumber\\ &\left.\int\left(\del q_t^{|T} - v_t^{|T}\del X^0_t - \sum_{i = 1}^kV_i(q_t^{|T})\del X^i_t\right)\right\rangle\\
    &= 0.\nonumber
\end{align}
We now prove the converse. Let $U\subseteq Q$ be open. Then $K^0 := U\times\R^n\times\R^n$ is an arbitrary chart on $\P Q$. It suffices to show that the stochastic implicit Euler-Lagrange equations are satisfied by $\Gamma_t^{|T} = (q_t^{|T}, v_t^{|T}, p_t^{|T})$ in $]]\hittingtime{K^0}, \hittingtime{K^0}+\hitexittime{K^0}[[$. Let $K$ be the closure of $K^0$ and $\epsilon\mapsto\Gamma_{\epsilon}$ be an arbitrary $(K, T)$-deformation of $\Gamma$. Then $\delta q_t = 0$ at $t = 0$ and $t = T$. Given a semimartingale $\Gamma_t = (q_t, v_t, p_t)$ on $\P Q$, define the $\R^{3(\dim Q)}$-valued semimartingale $\Xi(\Gamma)$ in local coordinates by 
\begin{align*}
    \Xi(q_t, v_t, p_t) &= \left(\int\left(\frac{\partial}{\partial q_t^{|T}}\left(\L\del X_t^0 + \sum_{i = 1}^k\left(L_i - \dotp{p_t^{|T}}{V_i(q_t^{|T})}\right)\circ dX^i_t\right) - \del \stopproc{p_t}{T}\right),\right.\\ &\left.\Sint{\left(p_t^{|T} - \pard{\L}{v_t^{|T}}\right)}{X^0_t},\int\left(\del q_t^{|T} - v_t^{|T}\del X^0_t - \sum_{i = 1}^kV_i(q_t^{|T})\del X^i_t\right)\right).
\end{align*}
Since the Stratonovich integral commutes with time changes we have $\Xi(\Gamma_{A_t}) = \Xi(\Gamma)_{A_t}$ for any continuous time change $t\mapsto A_t$. An application of Lemma \ref{kvariationlemma} shows that the integral from $0$ to $T$ in \eqref{calculationinlocalproof} can be replaced by an integral from $\hittingtime{K}$ to $\hittingtime{K}+\hitexittime{K}$. As a result, we obtain
\[\Sinttimedotp{\hittingtime{K}}{\hittingtime{K}+\hitexittime{K}}{\delta \Gamma^{|T}}{\Xi(\Gamma^{|T})}=0\]
for every $(K,T)$-deformation $\epsilon\mapsto \Gamma_{\epsilon}$ of $\Gamma$. By Lemma \ref{fundlem} $\del\Xi(\Gamma^{|T}) = 0$ in $]]\hittingtime{K^0}, \hittingtime{K^0}+\hitexittime{K^0}[[$. This implies that
\begin{align*}\label{sellocal}
    \circ dq_t^{|T} &= v_t^{|T}\circ dX^0_t + \sum_{i=1}^kV_i(q_t^{|T})\circ dX^i_t\nonumber\\
    \circ dp_t^{|T} &= \frac{\partial}{\partial q_t^{|T}}\left(\L\del X_t^0 + \sum_{i = 1}^k\left(L_i - \dotp{p_t^{|T}}{V_i(q_t^{|T})}\right)\circ dX^i_t\right)\nonumber\\
    \left(p_t^{|T} - \pard{\L}{v_t^{|T}}\right)\del X^0_t&=0\\
\end{align*}
in $]]\hittingtime{K^0}, \hittingtime{K^0}+\hitexittime{K^0}[[$. This completes the proof.
\end{proof}
\begin{remark}
    If $X^0 = t$ then the second equation is the Legendre transform $p_t = \pard{\L}{v_t}$. Also note that if $X^i = 0$ for all $i = 1, \cdots, k$ then the stochastic implicit Euler-Lagrange equations reduce to the deterministic implicit Euler-Lagrange equations.
\end{remark}
\begin{remark}
    Suppose $\L$ is a hyperregular Lagrangian, $F\L:TQ\rightarrow TQ$ is the Legendre transform of $\L$ and $E_{\L}(v_q):= \dotp{F\L(v_q)}{v_q} - v_q$ is the energy of $\L$. Define $\mathcal{H}\in\Cin(\cotQ)$ by $\mathcal{H} = E_{\L}\circ F\L\inv$ and $H_i\in \Cin(\cotQ)$ by $H_i(q,p) = \dotp{p}{V(q)} - L_i(q)$. Then, following the proof of Proposition 1 in Street and Takao \cite{street2023}, $\Gamma$ solves the stochastic Hamiltonian system
    \[\del \Gamma = X_{\mathcal{H}}(\Gamma)\del X^0 + \sum_{i = 1}^kX_{H_i}(\Gamma)\del X^i.\]
\end{remark}
\begin{remark}
    In Lázaro-Camí and Ortega \cite{lco3}, the authors develop a stochastic Hamilton-Jacobi equation for stochastic Hamiltonian systems using a stochastic version of Hamilton's principle in phase space. When $\L$ is a hyperregular Lagrangian we can pass to the Hamiltonian side as mentioned in the previous remark, and use the stochastic Hamilton-Jacobi formalism as proposed in \cite{lco3}. However, a general treatment of stochastic Hamilton-Jacobi equations on the Pontryagin bundle, similar to the deterministic development in Leok, Ohsawa and Sosa \cite{leokohsawasosa}, is left as a future prospect for research.   
\end{remark}

\subsection{A Stochastic Noether's Theorem}

Let $\Phi:Q\times \R\rightarrow Q$ be a deterministic smooth flow. Given $\epsilon\in\R$, let $\Phi_{\epsilon}:Q\rightarrow Q$ be the map $q\mapsto \Phi(q, \epsilon)$. Suppose $\L$ is invariant under the tangent lifted flow of $\Phi$, $L_i$ is invariant under $\Phi$ and $V_i$ is symmetric under $\Phi$, that is $V_i\circ \Phi_{\epsilon} = T\Phi_{\epsilon}\circ V_i$ for all $i = 1, \cdots, k$ and $\epsilon\in \R$. For every $\epsilon\in\R$, define $\Psi_{\epsilon}:\P Q\rightarrow \P Q$ by $\Psi_{\epsilon} = T\Phi_{\epsilon}\oplus T^*\Phi_{\epsilon}^{-1}$. Then $\pr{Q}\circ \Psi_{\epsilon} = \Phi_{\epsilon}$. 

\medskip

Let $\Gamma$ be a $\P Q$-valued admissible semimartingale, written in coordinates as $\Gamma_t = (q_t, v_t, p_t)$ and set $\Gamma_{\epsilon,t} = \Psi_{\epsilon}(\Gamma_t) = (q_{\epsilon,t},v_{\epsilon,t},p_{\epsilon,t})$. We note that:
\begin{align*}
        \int \G\del \Gamma_{\epsilon}&= \int \dotp{p_{\epsilon,t}}{\del q_{\epsilon,t}}\\
        &= \int \dotp{T^*\Phi_{\epsilon}\left(T^*\Phi_{\epsilon}\inv(p_t)\right)}{\del q_t}\\
        &= \int \dotp{p_t}{\del q_t}\\
        &= \int \G\del \Gamma
    \end{align*}
    and\begin{align*}
        E_j(\Gamma_{\epsilon,t}) &= E_j(q_{\epsilon,t}, v_{\epsilon,t}, p_{\epsilon,t}) = \begin{cases}
            \langle p_{\epsilon,t}, v_{\epsilon,t} \rangle - \L(q_{\epsilon,t}, v_{\epsilon,t}, p_{\epsilon,t}), & \text{if } j = 0, \\[8pt]
\langle p_{\epsilon,t}, V_j(q_{\epsilon,t})\rangle - L_j(q_{\epsilon,t}), & \text{if } j = 1, \dots, k.
        \end{cases}\\
        &= E_j(\Gamma_t).
    \end{align*}
Consequently, $\Ac{X}{\Gamma_{\epsilon}} = \Ac{X}{\Gamma}$ and hence $\D\Ac{X}{\Gamma} = 0$. Following the calculations done in the proof of Theorem \ref{localhamiltonpontryagin}, for $t\geq 0$ we have
\begin{align*}
    &\int_0^t\left\langle\delta \Gamma_s^{|t},\left(\int\left(\frac{\partial}{\partial q_s^{|t}}\left(\L\del X_s^0 +\right.\right.\right.\right.\\ &\left.\left.\sum_{i = 1}^k\left(L_i - \dotp{p_s^{|t}}{V_i(q_s^{|t})}\right)\circ dX^i_s\right) - \del \stopproc{p_s}{t}\right),\Sint{\left(p_s^{|t} - \pard{\L}{v_s^{|t}}\right)}{X^0_s},\\ &\left.\int\left(\del q_s^{|t} - v_s^{|t}\del X^0_s - \sum_{i = 1}^kV_i(q_s^{|t})\del X^i_s\right)\right\rangle \\&+ \dotp{p_t}{\delta q_t} - \dotp{p_0}{\delta q_0} = 0.
\end{align*}
Suppose $\Gamma$ solves the stochastic implicit Euler-Lagrange equations upto a maximal stopping time $\tau$. Then for all $t\leq \tau$, $\dotp{p_t}{\delta q_t} = \dotp{p_0}{\delta q_0}$. Hence $\dotp{p_t}{\delta q_t}$ is conserved along the solution $\Gamma$. Since $p_t = \pard{\L}{v_t}$, it follows that $\dotp{p_t}{\delta q_t} = \dotp{\theta_{\L}}{\delta q}$, where $\theta_{\L}$ is the pullback to $TQ$ of the Liouville 1-form on $\cotQ$ by $F\L$. Thus we have proven the following theorem:
\begin{theorem}
    Suppose $\Phi:Q\times\R\rightarrow Q$ is a smooth deterministic flow such that $\L$ is invariant under the tangent lift of $\Phi$, $L_i$ is invariant under $\Phi$ and $V_i$ is symmetric under $\Phi$ for all $i = 1, \cdots, k$. Let $\Gamma_t = (q_t, v_t, p_t)$ solves the stochastic implicit Euler-Lagrange equations and $q_{\epsilon,t} := \Phi_{\epsilon}(q_t) = \Phi_{\epsilon}(\pr{Q}(\Gamma_t))$. Then $\dotp{\theta_{\L}}{\delta q}$ is conserved along $\Gamma$.
\end{theorem}
\begin{remark}
    For a formulation of Noether's theorem using momentum maps, we refer to Lázaro-Camí and Ortega \cite{lco2}.
\end{remark}

\subsection{The Intrinsic Form of the Stochastic Hamilton-Pontryagin Principle}

We will now focus on the intrinsic form on the stochastic Hamilton-Pontryagin principle. We introduce an additional assumption here, namely that our semimartingales are obtained as solutions of Stratonovich equations on manifolds. To motivate this, we recall that in case of the deterministic Hamilton's principle, given a regular Lagrangian $\L\in\Cin(TQ)$ there exists a second order vector field $Z_{\L}$ such that the integral curves of $Z_{\L}$ project to solutions of the Euler-Lagrange equations (see Theorem 7.3.3 in Marsden and Ratiu \cite{marsden2}). Thus, for regular Lagrangians, finding a critical point of the action functional corresponds to selecting a particular vector field in $\VecF{TQ}$. We will extend this idea to the stochastic case as well, namely, we will show that under the assumption that a semimartingale solves a Stratonovich equation, finding a critical point of $\mathcal{S}_X$ corresponds to selecting a particular Stratonovich operator.

\medskip

Let $(e_0, \cdots, e_k)$ be a basis of $\R^{k+1}$. Let $\M$ be any regular submanifold of $\P Q$, $\Proj{T\M}$ and  $\Proj{\M}$ denote the same projections as $\Proj{T\P Q}$ and $\Proj{\P Q}$ respectively with $\P Q$ replaced by $\M$ and we restrict the other maps in \eqref{mapsinstrinsic} to $\M$. The generalized energies are now defined on $\M$ as opposed to $\P Q$. Given any Stratonovich operator $S\in\Strat{\R^{k+1}}{\M}$, any semimartingale $X=(X^0, \cdots, X^k)$ in $\R^{k+1}$ and a solution $\Gamma_X$ of $\del \Gamma = S(X,\Gamma)\del X$, from Eq. \eqref{hampontglobal} we have
\[\Ac{X}{\Gamma_X} = \Sinttime{0}{T}{\G}{\Gamma_X} - \sum_{j=0}^k\Sinttime{0}{T}{E_j(\Gamma_X)}{X^j}.\]
Given $x\in\R^{k+1}$, $y\in \M$ and $j \in \{0,\cdots, k\}$, suppose $z_j\in TT\M$, $\Proj{T\M} = S^{x,e_j}(y)$ and $T\Proj{\M} = w_y$, for some $w_y\in T\M$. Then \[dE_j(w_y) = dE_j\circ T\Proj{\M}(z_j) = (\Proj{\M})^*dE_j(S^{x,e_j})(z_j).\]
Now let $\epsilon\mapsto\Gamma_{X_\epsilon}$ be any deformation of $\Gamma_X$ and $Z^j$ be any $TT\M$-valued semimartingale such that $T\Proj{\M}(Z^j) = \delta \Gamma_X$ and $\Proj{T\M}(Z^j) = S^{x,e_j}(\Gamma_X)$. Then
\[dE_j(\delta\Gamma_X) = (\Proj{\M})^*dE_j(S^{X,e_j}(\Gamma_X))(Z^j)\]
and hence
\begin{equation}\label{variationinenergypart}
    \D\left(\sum_{j=0}^k\Sinttime{0}{T}{E_j(\Gamma_X)}{X^j}\right) = \sum_{j = 0}^k\Sinttime{0}{T}{(\Proj{\M})^*dE_j(S^{X,e_j}(\Gamma_X))(Z^j)}{X^j}.
\end{equation}
Next, by Lemma \ref{variationsofintegrals} 
\begin{align*}\D\Sinttime{0}{T}{\G}{\Gamma_X}&= \Sinttime{0}{T}{i_{\delta\Gamma_X}\d\G}{\Gamma_X} + \dotp{\G(\Gamma_{X_T})}{\delta\Gamma_{X_T}} - \dotp{\G(\Gamma_{X_0})}{\delta\Gamma_{X_0}}\\
&= \Sinttime{0}{T}{\Dual{S}(X,\Gamma_X)i_{\delta\Gamma_X}\d\G}{X} + \dotp{\G(\Gamma_{X_T})}{\delta\Gamma_{X_T}} - \dotp{\G(\Gamma_{X_0})}{\delta\Gamma_{X_0}}.
\end{align*}
The calculation of $\int_0^T S^{\vee}(X,\Gamma_X)i_{\delta\Gamma_X}\d\mathcal{G}\del X$ is related to the proof of Proposition 3.2 in Yoshimura and Marsden \cite{yoshimura06}. Let $\theta_{\cotQ}$ denote the Liouville 1-form on $\cotQ$ and $\Omega_{\cotQ} = -d\theta_{\cotQ}$. Denote by \[\Omega_{\cotQ}^{\flat}:T\cotQ\rightarrow T^*\cotQ\] the bundle map associated with $\Omega_{\cotQ}$. Also let $\theta_{T^*\cotQ}$ be the Liouville 1-form on $T^*\cotQ$ and set $\chi = (\Omega^{\flat})^*\theta_{T^*\cotQ}$. Note that $\chi$ is a 1-form on $T\cotQ$. We will show that, given any $x\in\R^{k+1}$, $y = (q,v,p)\in\M$, $w_y = (q,v,p,w_q,w_v,w_p)\in T\M$ and $z_j\in TT\M$ such that $T\Proj{\M}(z_j) = w_y$ and $\Proj{T\M}(z_j) = S^{x,e_j}(y)$, we have
\begin{equation}\label{importantequationforproof}
    \dotp{\Dual{S}(x,y)i_{w_y}\d\G}{e_j} = (T\pr{\cotQ})^*\chi(S^{x,e_j}(y))(z_j).
\end{equation}
Let
\begin{align*}
    u_y &= S^{x,e_j}(y) = (q,v,p,u_q,u_v,u_p)\in T\M,\\
    z_j &= (q,v,p,u_q,u_v,u_p, w_q,w_v,w_p, \tilde{w}_q, \tilde{w}_v,\tilde{w}_p).
\end{align*}
From the proof of Lemma \ref{variationinG} we get
\[\dotp{\Dual{S}(x,y)i_{w_y}\d\G}{e_j} = \d\G(w_y,u_y) = \dotp{w_p}{u_q} - \dotp{w_q}{u_p}.\]
On the other hand
\begin{align*}
    \theta_{T^*\cotQ}(\Omega_{\cotQ}^{\flat}\circ T\pr{\cotQ}(q,v,p,u_q,u_v,u_p)) &= \theta_{T^*\cotQ}(\Omega_{\cotQ}^{\flat}(q,p,u_q,u_p))\\
    &= \theta_{T^*\cotQ}(q,p,-u_p, u_q)\\
    &= -u_pdq+u_qdp.
\end{align*}
Since $T_{u_y}T\pr{\cotQ}(z_j) = (q,p,u_q,u_p,w_q,w_p,\tilde{w}_q,\tilde{w}_p)$, it follows that
\[T\Omega_{\cotQ}^{\flat}(T_{u_y}T\pr{\cotQ})(z_j) = (q,p,-u_p,u_q,w_q,w_p,-\tilde{w}_p, \tilde{w}_q)\]
Note that 
\begin{align*}
    (T\pr{\cotQ})^*\chi(S^{x,e_j}(y))(z_j)&= \theta_{T^*\cotQ}(\Omega_{\cotQ}^{\flat}\circ T\pr{\cotQ}(u_y))\cdot T\Omega_{\cotQ}^{\flat}(T_{u_y}T\pr{\cotQ}(z_j)),
\end{align*}
and
\begin{align*}
\theta_{T^*\cotQ}(\Omega_{\cotQ}^{\flat}\circ T\pr{\cotQ}(u_y))\cdot T\Omega_{\cotQ}^{\flat}(T_{u_y}T\pr{\cotQ}(z_j)) &=- \dotp{w_q}{u_p}+\dotp{w_p}{u_q}\\ &= \dotp{\Dual{S}(x,y)i_{w_y}\d\G}{e_j}.\end{align*}
This proves our claim. 

\medskip

Consequently, given any $a = (a^0, \cdots, a^k)\in\R^{k+1}$
\begin{align*}
    \langle S^{\vee}(x,y)i_{w_y}\d\mathcal{G}, a\rangle &= 
    \sum_{j = 0}^k a^j \langle S^{\vee}(x,y)i_{w_y}\d\mathcal{G},e_j\rangle\\
    &= \sum_{j = 0}^k a^j(T\pr{\cotQ})^*\chi(S^{x,e_j}(y))(z_j)
\end{align*}
where $z_j\in TT\M$, $\Proj{T\M}(z_j) = S^{x,e_j}(y)$ and $T\Proj{\M}(z_j) = w_y$.

\medskip

Therefore, given any deformation $\epsilon\mapsto \Gamma_{X_\epsilon}$ and semimartingales $Z^0,\cdots, Z^k$ in $TT\M$ over $S^{X,e_j}(\Gamma_X)$ such that $T\Proj{\M}(Z^j) = \delta \Gamma_X$ we have
\begin{align*}
    \D\Sinttime{0}{T}{\G}{\Gamma_X}&=
\Sinttime{0}{T}{\Dual{S}(X,\Gamma_X)i_{\delta\Gamma_X}\d\G}{X} + \dotp{\G(\Gamma_{X_T})}{\delta\Gamma_{X_T}} - \dotp{\G(\Gamma_{X_0})}{\delta\Gamma_{X_0}}\\
&= \sum_{j = 0}^k \Sinttime{0}{T}{(T\pr{\cotQ})^*\chi(S^{X,e_j}(\Gamma_X))(Z^j)}{X^j}\\&+\dotp{\G(\Gamma_{X_T})}{\delta\Gamma_{X_T}} - \dotp{\G(\Gamma_{X_0})}{\delta\Gamma_{X_0}}.
\end{align*}
We have $\langle \mathcal{G}(\Gamma_t), \delta \Gamma_t\rangle = \dotp{p_t} {\delta q_t}$, where $\Gamma_t = (q_t, v_t, p_t)$ in coordinates. Assuming that $T\pr{Q}(\delta \Gamma) = 0$ at $t = 0,T$, we have \[\left\langle \mathcal{G}(\Gamma_T),\delta \Gamma_T\right\rangle = 0 = \left\langle\mathcal{G}(\Gamma_0),\delta \Gamma_0\right\rangle.\]
In that case
\[\D\Sinttime{0}{T}{\G}{\Gamma_X} =\sum_{j = 0}^k \Sinttime{0}{T}{(T\pr{\cotQ})^*\chi(S^{X,e_j}(\Gamma_X))(Z^j)}{X^j}.\]
We summarize this discussion in the following lemma:
\begin{lemma}\label{importantlemmaintrinsic}
    Let $S\in\Strat{\R^{k+1}}{\M}$ where $\M$ is a regular submanifold of $\P Q$. For every semimartingale $X = (X^0, \cdots, X^k)\in \sem{\R^{k+1}}$, if $\Gamma_X$ solves $\del \Gamma = S(X,\Gamma)\del X$ and $\Gamma_X$ is admissible, then
    \[\D\Ac{X}{\Gamma_X} = \D\left[\Sinttime{0}{T}{\G}{\Gamma_X} - \sum_{j= 0}^k\Sinttime{0}{T}{E_j(\Gamma_X)}{X^j}\right] = 0\]
    for all variations $\epsilon \mapsto \Gamma_{X_\epsilon}$ with $T\pr{Q}(\delta\Gamma_{X_0}) = T\pr{Q}(\delta\Gamma_{X_T}) = 0$, if and only if 
    \begin{equation}\label{globalvariationofaction}\sum_{j=0}^k\Sinttime{0}{T}{\left[(T\pr{\cotQ})^*\chi(S^{X,e_j}(\Gamma_X)) - (\Proj{\M})^*dE_j(S^{X,e_j}(\Gamma_X))\right](Z^j)}{X^j} = 0\end{equation}
    for arbitrary $TT\M$-valued semimartingales $Z^0,\cdots, Z^k$ over $S^{X,e_j}(\Gamma_X)$ such that $T\left(\pr{Q}\circ \Proj{\M}\right)(Z^j) = 0$ at $t = 0$ and $t = T$.
\end{lemma}
\begin{definition}
    Given $y_0\in \M$ and $j = 0, \cdots, k$ let $y_j(t)$ be a curve in $\M$ satisfying
    \begin{equation}\label{deterministicHamiltonPontryagin}
        (T\pr{\cotQ})^*\chi(y_j(t),\dot{y}_j(t)) = (\Proj{\M})^*dE_j(y_j(t), \dot{y}_j(t))
    \end{equation}
    with $y_j(0) = y_0$. Let $S_{HP}\in\Strat{\R^{k+1}}{\M}$ be defined by 
    \[S_{HP}(x_0,y_0)(a) = \sum_{j = 0}^k a^j\dot{y}_j(0)\in T_{y_0}\M\]
    for every $x_0\in\R^{k+1}$ and $a=(a^0, \cdots, a^k)\in \R^{k+1}\cong T_{x_0}\R^{k+1}$. We will call $S_{HP}$ a \textbf{Hamilton-Pontryagin Stratonovich operator}.  
\end{definition}
\begin{remark}
    From Yoshimura and Marsden \cite{yoshimura06}, if $j = 0$ then Eq. \eqref{deterministicHamiltonPontryagin} is just the deterministic implicit Euler-Lagrange equations for $\L$. Since this means that the Legendre transform holds, it follows that $S_{HP}^{x,e_0}$, and hence $S_{HP}$, is well-defined if and only if $\M$ is the submanifold $\mathcal{K} = TQ\oplus F\L(TQ)$. For $j = 1, \cdots, k$, in local coordinates Eq. \eqref{deterministicHamiltonPontryagin} reads
    \begin{equation}\label{SHPi}\dot{q} = V_j(q),\:\:\:\dot{p} = \pard{}{q}\left(L_j(q) - \dotp{p}{V_j(q)}\right).\end{equation}
\end{remark}
By definition of $S_{HP}$, if $\M = \mathcal{K}$ and $S = S_{HP}$ then Eq. \eqref{globalvariationofaction} is satisfied, so $\D\Ac{X}{\Gamma_X} = 0$. We now prove the converse.

\medskip

First we make an important observation. The Stratonovich operator $S$ is a deterministic object that is defined independently of the semimartingale $X$. Since the equivalence in Lemma \ref{importantlemmaintrinsic} is true for every semimartingale $X\in\sem{\R^{k+1}}$ and a solution $\Gamma_X$ of $\del\Gamma = S(X,\Gamma)\del X$ that is admissible, it must be true for a deterministic semimartingale of the form $X_t = X_0(t):= x_0+te_0$, where $x_0\in\R^{k+1}$ is arbitrary. In this case the solution of $\del\Gamma = S(X,\Gamma)\del X = S^{X_0(t),e_0}(\Gamma)dt$ is a deterministic smooth curve in $\M$ that we denote by $\gamma_0(t)$. Given $y_0\in \M$, suppose that $\gamma_0(0)= y_0$.  Note that $\gamma_0(t)$ solves \[\dot{\gamma}_{0}(t) = S^{X_0(t),e_0}(\gamma_0(t)) = S(x_0+te_0,\gamma_0(t))(e_0).\]
By Lemma \ref{importantlemmaintrinsic} we have
\[\int_0^T{\left[(T\pr{\cotQ})^*\chi(S^{x_0+te_0,e_0}(\gamma_0(t))) - (\Proj{\M})^*dE_0(S^{x_0+te_0,e_0}(\gamma_j(t)))\right](z^0(t))}dt = 0\]
for all smooth curves $z_0(t)$ in $TT\M$ such that $\Proj{T\M}(z_0(t)) = S^{x_0+te_0,e_0}(\gamma_0(t)) = \dot{\gamma}_0(t)$ and $T\left(\pr{Q}\circ \Proj{\M}\right)(z_0(t)) = 0$ at $t = 0$ and $t = T$. By the (deterministic) fundamental theorem of the calculus of variations we have
\[(T\pr{\cotQ})^*\chi(S^{x_0+te_0,e_0}(\gamma_0(t))) = (\Proj{\M})^*dE_0(S^{x_0+te_0,e_0}(\gamma_0(t))).\]
Since $\dot{\gamma}_{0}(t) = S^{X_0(t),e_0}(\gamma_0(t))$ we have
\[(T\pr{\cotQ})^*\chi(\gamma_0(t),\dot{\gamma}_0(t)) = (\Proj{\M})^*dE_0(\gamma_0(t),\dot{\gamma}_0(t)).\]
This equation has a solution provided $\M = \mathcal{K}$. Moreover, by definition of $S_{HP}$
\[S_{HP}(x_0,y_0)(e_0) = \dot{\gamma}_0(0) = S(x_0,y_0)(e_0).\]
A similar argument with $e_0$ replaced by $e_j$ shows that $S(x_0, y_0)(e_j) = S_{HP}(x_0, y_0)(e_j)$ for all $j = 0, \cdots, k$. Thus $S = S_{HP}$. 
\begin{theorem}
    Let $S\in\Strat{\R^{k+1}}{\M}$ where $\M$ is a regular submanifold of $\P Q$. For every semimartingale $X = (X^0, \cdots, X^k)\in\sem{\R^{k+1}}$, if $\Gamma_X$ solves $\del \Gamma = S(X, \Gamma)\del X$ and $\Gamma_X$ is admissible then $\D\Ac{X}{\Gamma_X} = 0$ for all deformations $\epsilon\mapsto \Gamma_{X_\epsilon}$ satisfying $T\pr{Q}(\delta {\Gamma_{X_0}}) = T\pr{Q}(\delta {\Gamma_{X_T}}) = 0$ if and only if $\M = \mathcal{K}$ and $S = S_{HP}$.

    \medskip
    
Moreover, the stochastic implicit Euler-Lagrange equations are given by \[\del (\pr{\cotQ}(\Gamma^{|T})) = T\pr{\cotQ}\left(S_{HP}(X, \Gamma^{|T})\right)\del X\] on the submanifold $\mathcal{K}$.
\end{theorem}
\begin{proof}
   It only remains to show that the stochastic implicit Euler-Lagrange equations are given by\[\del (\pr{\cotQ}(\Gamma^{|T})) = T\pr{\cotQ}\left(S_{HP}(X, \Gamma^{|T})\right)\del X\] on $\mathcal{K}$. Let $\Gamma_t = (q_t,v_t,p_t)$ in terms of local coordinates on $\P Q$. The restriction to $\mathcal{K}$ implies that $\left(p_t^{|T} - \pard{\L}{v_t^{|T}}\right)\del X^0_t = 0$. Given any $x\in\R^{k+1}$ and $a = (a^0, \cdots, a^k)\in \R^{k+1}\cong T_x\R^{k+1}$, suppose that in local coordinates, we have \[S_{HP}(x,y)(a) = \sum_{j = 0}^ka^jS_{HP}^{x, e_j}(q,v,p) = \sum_{j = 0}^ka^j(u_{q_j}, u_{v_j}, u_{p_j}).\]Using the local form of the deterministic implicit Euler-Lagrange equations for $j = 0$ and using Eq. \eqref{SHPi} for $j = 1, \cdots, k$, we have
   \begin{align*}T\pr{\cotQ}\left(S_{HP}(x, y)(a)\right) &= \sum_{j = 0}^ka^jT\pr{\cotQ}(u_{q_j}, u_{v_j}, u_{p_j})\\ &= \sum_{j = 0}^ka^j(u_{q_j}, u_{p_j})\\ &= a^0\left(v, \pard{\L}{q}\right) + \sum_{j = 1}^ka^j\left(V_j(q), \pard{}{q}\left(L_j(q) - \dotp{p}{V_j(q)}\right)\right). \end{align*}Thus the local form of the equation $\del (\pr{\cotQ}(\Gamma^{|T})) = T\pr{\cotQ}\left(S_{HP}(X, \Gamma^{|T})\right)\del X$ is
   \begin{align*}
       \circ dq_t^{|T} &= v_t^{|T}\circ dX^0_t + \sum_{i=1}^kV_i(q_t^{|T})\circ dX^i_t\\
    \circ dp_t^{|T} &= \frac{\partial}{\partial q_t^{|T}}\left(\L\del X_t^0 + \sum_{i = 1}^k\left(L_i(q^{|T}) - \dotp{p_t^{|T}}{V_i(q_t^{|T})}\right)\circ dX^i_t\right),
   \end{align*}
   which, together with the equation $\left(p_t^{|T} - \pard{\L}{v_t^{|T}}\right)\del X^0_t = 0$ gives the stochastic implicit Euler-Lagrange equations.
\end{proof}
\begin{remark}
    This method of reformulating the problem of determining a critical point of a stochastic action to determining a Stratonovich operator can be applied to other stochastic action principles as well, for instance, the stochastic Hamilton's principle in phase space described in Lázaro-Camí and Ortega \cite{lco1}.
\end{remark}

\section*{Acknowledgments}
The author is grateful to Cristina Stoica, Tanya Schmah, Ana Bela Cruzeiro and Jean-Claude Zambrini for valuable comments and insights. The author is also thankful to the anonymous referees for their suggestions. This work is supported by a Ph.D. scholarship from the University of Ottawa.


\begin{thebibliography}{0}
\bibitem{street2023}
O. D. Street and S. Takao, Semimartingale driven mechanics and reduction by symmetry for stochastic and dissipative dynamical systems, {\it arXiv preprint}, {\bf 2312.09769} (2023), \url{https://arxiv.org/abs/2312.09769}.

\bibitem{lco1}
J.-A. Lázaro-Camí and J.-P. Ortega, Stochastic {H}amiltonian dynamical systems, {\it Reports on Mathematical Physics}, {\bf 61} (2008), 65--122, \url{https://doi.org/10.1016/S0034-4877(08)80003-1}.

\bibitem{lco2}
J.-A. Lázaro-Camí and J.-P. Ortega, Reduction, reconstruction, and skew-product decomposition of symmetric stochastic differential equations, {\it Stochastics and Dynamics}, {\bf 9} (2009), 1.

\bibitem{lco3}
J.-A. Lázaro-Camí and J.-P. Ortega, The stochastic {H}amilton-{J}acobi equation, {\it Journal of Geometric Mechanics}, {\bf 1} (2009), 3, 295–315.


\bibitem{cipriano2007}
F. Cipriano and A. B. Cruzeiro, Navier-Stokes equation and diffusions on the group of homeomorphisms of the torus, {\it Communications in Mathematical Physics}, {\bf 275} (2007), 255--269.

\bibitem{arnaudon1998}
M. Arnaudon and A. Thalmaier, Stability of stochastic differential equations in manifolds, in {\it Séminaire de Probabilités XXXII}, Springer, Berlin, Heidelberg, {\bf 188} (1998), 188--214.

\bibitem{bou2009stochastic}
N. Bou-Rabee and H. Owhadi, Stochastic variational integrators, {\it IMA Journal of Numerical Analysis}, {\bf 29} (2008), 421--443, \url{http://dx.doi.org/10.1093/imanum/drn018}.

\bibitem{holm2}
A. Arnaudon, A. L. De Castro, and D. D. Holm, Noise and dissipation on coadjoint orbits, {\it Journal of Nonlinear Science}, {\bf 28} (2016), 91--145.

\bibitem{holm3}
A. B. Cruzeiro, D. D. Holm, and T. S. Ratiu, Momentum maps and stochastic Clebsch action principles, {\it Communications in Mathematical Physics}, {\bf 357} (2016), 873--912.

\bibitem{yoshimura061}
H. Yoshimura and J. E. Marsden, Dirac structures in Lagrangian mechanics Part I: Implicit Lagrangian systems, {\it Journal of Geometry and Physics}, {\bf 57} (2006), 133--156, \url{https://doi.org/10.1016/j.geomphys.2006.02.009}.

\bibitem{yoshimura06}
H. Yoshimura and J. E. Marsden, Dirac structures in Lagrangian mechanics Part II: Variational structures, {\it Journal of Geometry and Physics}, {\bf 57} (2006), 209--250, \url{https://doi.org/10.1016/j.geomphys.2006.02.012}.

\bibitem{marsden2}
J. E. Marsden and T. S. Ratiu, {\it Introduction to Mechanics and Symmetry: A Basic Exposition of Classical Mechanical Systems}, 2nd edition, Springer, New York, 1999.

\bibitem{emerybook}
M. Emery, {\it Stochastic Calculus in Manifolds}, Springer, Berlin Heidelberg, 2012.

\bibitem{nelson}
E. Nelson, {\it Dynamical Theories of Brownian Motion}, Princeton University Press, Princeton, N.J., 1967.

\bibitem{nelson2}
E. Nelson, Derivation of the Schrödinger equation from Newtonian mechanics, {\it Physical Review}, {\bf 150} (1966), 1079--1085, \url{https://doi.org/10.1103/PhysRev.150.1079}.

\bibitem{yasue}
K. Yasue, Stochastic calculus of variations, {\it Journal of Functional Analysis}, {\bf 41} (1981), 327--340, \url{https://doi.org/10.1016/0022-1236(81)90079-3}.

\bibitem{anabela}
M. Arnaudon and A. B. Cruzeiro, Lagrangian Navier--Stokes diffusions on manifolds: Variational principle and stability, {\it Bulletin des Sciences Mathématiques}, {\bf 136} (2012), 857--881, \url{https://doi.org/10.1016/j.bulsci.2012.06.007}.

\bibitem{zambrini1986variational}
J.-C. Zambrini, Variational processes and stochastic versions of mechanics, {\it Journal of Mathematical Physics}, {\bf 27} (1986), 2307--2330.

\bibitem{arn}
M. Arnaudon, X. Chen, and A. B. Cruzeiro, Stochastic Euler--Poincaré reduction, {\it Journal of Mathematical Physics}, {\bf 55} (2014), 081507, \url{https://doi.org/10.1063/1.4893357}.

\bibitem{holm1}
D. D. Holm, Variational principles for stochastic fluid dynamics, {\it Proceedings of the Royal Society A}, {\bf 471} (2015), 20140963, \url{https://doi.org/10.1098/rspa.2014.0963}.

\bibitem{crisan}
O. D. Street and D. Crisan, Semi-martingale driven variational principles, {\it Proceedings of the Royal Society A}, {\bf 477} (2021), 20200957, \url{https://doi.org/10.1098/rspa.2020.0957}.

\bibitem{huangzambrini}
Q. Huang and J.-C. Zambrini, From second-order differential geometry to stochastic geometric mechanics, {\it Journal of Nonlinear Science}, {\bf 33} (2023), 4.

\bibitem{leokohsawasosa}
M. Leok, T. Ohsawa, and D. Sosa, Hamilton–Jacobi theory for degenerate Lagrangian systems with holonomic and nonholonomic constraints, {\it Journal of Mathematical Physics}, {\bf 53} (2012), no. 7.


\bibitem{zambrini2014researchprogramstochasticdeformation}
J.-C. Zambrini, {\it The research program of Stochastic Deformation (with a view toward Geometric Mechanics)}, arXiv preprint, 2014, \url{https://arxiv.org/abs/1212.4186}.

\end{thebibliography}
\end{document}